%% file: ms.tex
\documentclass[10pt,twocolumn,letterpaper]{article}

\usepackage{preamble}

\begin{document}

\title{\sys: Oblivious Serializable Transactions in the Cloud}
\author{
  \begin{tabular}{c@{\hspace{3em}}c@{\hspace{3em}}c}
    Natacha Crooks$^\star$\tinydag & Matthew Burke\tinydag & Ethan Cecchetti\tinydag \\[0.25em]
    Sitar Harel\tinydag & Rachit Agarwal\tinydag & Lorenzo Alvisi\tinydag
  \end{tabular}\\[1.5em]
  $^\star$University of Texas at Austin \hspace{3em} \tinydag{}Cornell University}
\date{}
\maketitle
\thispagestyle{empty}

\input{abstract-current}
\input{intro-current}

\input{threat-current}

\input{background-current}

\input{design-current}

\input{proxy-design-current}

\input{oram-design-current}

\input{durability-design-current}

\input{security-current}

\input{implementation-current}
\input{evaluation-current}

\input{related-current}

\input{conclusion-current}

{\footnotesize \bibliographystyle{acm}
\bibliography{refs}}

\appendix
\input{integrity-current}

\input{formal-security-current}

\end{document}

%% file: abstract-current.tex
\begin{abstract}
    \vspace{5pt}
    \par This paper presents the design and implementation of \sys{},
    the first system to provide ACID transactions while also hiding
    access patterns. \sys{} uses as its building block oblivious RAM,
    but turns the demands of supporting transactions into a
    performance opportunity. By executing transactions within epochs
    and delaying commit decisions until an epoch ends, \sys{} reduces
    the amortized bandwidth costs of oblivious storage and increases
    overall system throughput. These performance gains, combined with
    new oblivious mechanisms for concurrency control and recovery,
    allow \sys{} to execute OLTP workloads with reasonable throughput:
    it comes within 5\texttimes\ to 12\texttimes\ of a non-oblivious
    baseline on the TPC-C, SmallBank, and FreeHealth
    applications. Latency overheads, however, are higher
    (70\texttimes\ on TPC-C).
\end{abstract}

%% file: intro-current.tex
\section{Introduction} 
This paper presents 
\sys{}, the first cloud-based key value store that supports transactions 
while hiding access patterns from cloud providers.
\sys{} aims to mitigate the fundamental tension between the
convenience of offloading data to the cloud, and the significant
privacy concerns that doing so creates. On the one hand, cloud
services~\cite{s3,simpledb,azuretables,documentdb,clouddatastore}
offer clients scalable, reliable IT solutions and present application
developers with feature-rich environments (transactional support,
stronger consistency guarantees~\cite{dynamodb,mongodb}, etc.).
Medical practices, for instance, increasingly prefer to use
cloud-based software to manage electronic health records
(EHR)~\cite{carecloud,kuo2011opportunities}.  On the other hand, many
applications that could benefit from cloud services store personal
data that can reveal sensitive information even when encrypted or
anonymized~\cite{narayanan08robust,singel09netflix,narayanan10myths,wang17splinter}.
For example, charts accessed by oncologists can reveal not only
whether a patient has cancer, but also, depending on the frequency of
accesses (e.g., the frequency of chemotherapy appointments),
indicate the cancer's type and severity.  Similarly, travel agency
websites have been suspected of increasing the price of frequently
searched flights~\cite{wang17splinter}. Hiding \textit{access
  patterns}---that is, hiding not only the content of an object, but
also when and how frequently it is accessed, is thus often desirable.

Responding to this challenge, the systems community has taken a
fresh look at  private data
access. Recent solutions, whether based on private information
retrieval~\cite{gupta16popcorn,melchor14xpir}, Oblivious
RAM~\cite{sahin2016taostore,lorch2013shroud,cecchetti2017solidus},
function sharing~\cite{wang17splinter}, or trusted
hardware~\cite{saba2017oblidb,arasu2013cipherbase,bajaj2011trusteddb,lorch2013shroud,tu2013monomi},
show that it is possible to support complex SQL queries without
revealing access patterns.

\sys addresses a complementary issue: supporting ACID transactions
while guaranteeing data access privacy.  This combination raises
unique challenges~\cite{arasu2013cipherbase}, as
concurrency control mechanisms used to enforce
isolation, and techniques used
to enforce atomicity and durability, all
make hiding access patterns more problematic (\S\ref{sec:threat}).

\sys takes as its starting point Oblivious RAM, which provably hides
all access patterns.  Existing ORAM implementations, however, cannot
support transactions. First, they are not fault-tolerant. For security
and performance, they often store data in a client-side
\textit{stash}; durability requires the stash content to be
recoverable after a failure, and preserving privacy demands hiding the
stash's size and contents, even during failure recovery.  Second, ORAM
provides limited or no support for
concurrency~\cite{boyle16pram,williams12privatefs,stefanov2013oblivistore,sahin2016taostore},
while transactional systems are expected to sustain highly concurrent
loads.

\sys demonstrates that the demands of supporting transactions can
not only be met, but also turned into a performance opportunity. Its
key insight is that  transactions actually afford more 
flexibility than the single-value operations supported by previous
ORAMs. For example, serializability~\cite{Papadimitriou1979serializability}
requires that the effects of transactions be reflected consistently in
the state of the database \textit{only after they commit}.  \sys
leverages this flexibility to delay committing transactions until the
end of fixed-size epochs, buffering their execution at a trusted proxy
and enforcing consistency and durability only at epoch boundaries.
This delay improves ORAM throughput without weakening
privacy.

The ethos of \textit{delayed visibility} is the core that drives
\sys's design. First, it allows \sys to implement a multiversioned
database atop a single-versioned ORAM, so that read operations proceed
without blocking, as with other multiversioned
databases~\cite{bernstein1983mcc}, and intermediate writes are
buffered locally: only the \textit{last} value of any key modified
during an epoch is written back to the ORAM. Delaying writes reduces
the number of ORAM operations needed to commit a transaction, lowering
amortized CPU and bandwidth costs without increasing contention:
\sys's concurrency control ensures that delaying commits does not
affect the set of values that transactions executing within the same
epoch can observe.

Second, it allows \sys{} to securely parallelize Ring
ORAM~\cite{ren15ringoram}, the ORAM construction on which it 
builds. \sys{} pipelines conflicting ORAM operations
rather than processing them sequentially, as existing ORAM
implementations do. This parallelization, however, is only secure if
the write-back phase of the ORAM algorithm is delayed until
pre-determined times, namely, epoch boundaries.

Finally, delaying visibility gives \sys{} the ability to abort entire
epochs in case of failure. \sys{} leverages this flexibility, along
with the near-deterministic write-back algorithm used by Ring ORAM, to
drastically reduce the information that must be logged to guarantee
durability and privacy-preserving crash recovery.

The results of a prototype implementation of \sys{} are promising. On
three applications (TPC-C~\cite{tpcc},
SmallBank~\cite{difallah2013oltpbench}, and
FreeHealth~\cite{freehealthfrench}, a real medical application) \sys{}
is within 5\texttimes-12\texttimes\ of the throughput of non-private
baselines.  Latency is higher (70\texttimes), but remains reasonable
(in the hundreds of milliseconds).

To summarize, this paper makes three
contributions:
\begin{enumerate}
\item It presents the design, implementation, and evaluation of the
  first ACID transactional system that also hides access patterns.
\item It introduces an epoch-based design that leverages the
  flexibility of transactional workloads to increase overall system
  throughput and efficiently recover from failures.
\item It provides the first formal security definition of a
  transactional, crash-prone, and private database. \sys{} uses the
  UC-security framework~\cite{canetti01}, ensuring that security
  guarantees hold under concurrency and composition.
\end{enumerate}

\sys{} also has several limitations.  First, like most ORAMs
that regulate the interactions of multiple clients, it relies on a
local centralized proxy, which introduces issues of 
  fault-tolerance and scalability.
   Second, \sys{} does
not currently support range or complex SQL queries.
Addressing the consistency challenge of maintaining
  oblivious
  indices~\cite{arasu2013cipherbase,saba2017oblidb,zheng17opaque} in
  the presence of transactions is a promising avenue for future
  work.

%% file: threat-current.tex
\section{Threat and Failure Model}
\label{sec:threat}

\sys's threat and failure assumptions aim to model deployments
similar to those of medical practices, where  doctors and
nurses access medical records through an on-site server, but choose to outsource the integrity and availability of
those records to a cloud storage service~\cite{carecloud,kuo2011opportunities}.

\par\textbf{Threat Model}. \sys{} adopts a \textit{trusted
  proxy} threat
model~\cite{stefanov2013oblivistore,sahin2016taostore,williams12privatefs}:
it assumes multiple mutually-trusting client applications interacting
with a single trusted proxy in a single shared administrative domain.
The applications issue transactions and the proxy manages their
execution, sending read and write requests on their behalf over an
asynchronous and unreliable network to an {\em untrusted storage
  server}.
This server is controlled by an honest-but-curious adversary that can
observe and control the timing of communication to and from the proxy,
but not the on-site communication between application clients and the
proxy. We extend our threat model to a fully
malicious adversary in Appendix~\ref{sec:integrity}.
We consider attacks that leak information by
  exploiting timing channel vulnerabilities in modern
  processors~\cite{Lipp2018meltdown,Kocher2018spectre,bulck2018foreshadow} to be out of scope. \sys
guarantees that the adversary will learn no information about:
\one~the decision (commit/abort) of any ongoing transaction; \two~the
number of operations in an ongoing transaction; \three~the type of
requests issued to the server; and \four~the actual data they
access.  \sys does not seek to hide the type of application that is
currently executing (ex: OLTP vs OLAP).

    \par \textbf {Failure Model}. \sys{} assumes cloud storage is 
    reliable, but, unlike previous ORAMs,
    explicitly considers that both application clients and the proxy
    may fail. These failures should  invalidate neither \sys's privacy
    guarantees nor the Durability and Atomicity of transactions.

 \section{Towards Private Transactions}
\label{sec:threat}

Many distributed, disk-based commercial database systems~\cite{mysql,corbett2013spanner,baker11megastore} separate concurrency
control logic from storage management: SQL queries and transactional requests are
regulated in a concurrency control unit and are subsequently converted to 
simple read-write accesses to key-value/file system storage. As ORAMs expose a read-write address
space to users, a logical first attempt at implementing
oblivious transactions would simply replace the database storage with an arbitrary ORAM.
This black-box approach, however, raises both security concerns
(\S\ref{subsec:blackboxsecurity}) and performance/functionality issues (\S\ref{subsec:performance})

Security guarantees can be compromised by simply enforcing
the ACID properties. Ensuring Atomicity, Isolation, and Durability
imposes additional structure on the order of individual reads and
writes, introducing sources of information
leakage~\cite{sheff2016sss,arasu2013cipherbase} that do not exist in
non-transactional ORAMs (\S\ref{subsec:blackboxsecurity}).
Performance and functionality, on the other hand, are hampered by the
inability of current ORAMs to efficiently support highly concurrent
loads and guarantee Durability.

\subsection{Security for Isolation and Durability}
\label{subsec:blackboxsecurity}

The mechanisms used to guarantee Isolation, Atomicity, and Durability introduce
timing correlations that directly leak information about the data accessed
by ongoing transactions.

\par \textbf{Concurrency Control}.  Pessimistic concurrency controls
like two-phase locking~\cite{eswaran1976ncp} delay operations that
would violate serializability: a write operation from transaction
$T_1$ cannot execute concurrently with any operation to the same
object from transaction $T_2$. Such blocking can potentially reveal
sensitive information about the data, even when executing on top of a
construction that hides access patterns: a sudden drop in throughput
could reveal the presence of a deadlock, of a write-heavy transaction
blocking the progress of read transactions, or of highly contended
items accessed by many concurrent transactions.  More aggressive
concurrency control schemes like timestamp ordering or multiversioned
concurrency control%
~\cite{bernstein1983mcc,reed1979tso,larson2011hekaton,
  agrawal1990locks,jones10low,reddy2004speculative,xie2015callas}
allow transactions to observe the result of the writes of other
ongoing transactions. These schemes improve performance in contended
workloads, but introduce the potential for {\em cascading aborts}: if
a transaction aborts, all transactions that observed its write must
also abort. If a write-heavy transaction $T_{\it heavy}$ aborts, it
may cause a large number of transactions to rollback, again revealing
information about $T_{\it heavy}$ and, perhaps more
  problematically, about the set of objects that $T_{\it heavy}$
  accessed.
\par \textbf{Failure Recovery}.
When recovering from failure, Durability requires preserving the effects of committed transactions, while
Atomicity demands removing any changes caused by partially-executed transactions. 
Most commercial systems~\cite{
  mysqlcluster,mysql,sqlserver} preserve these
properties through variants of \textit{undo} and \textit {redo}
logging.
To guarantee Durability,
write and commit operations
are written to a redo log that is replayed after a failure. To
guarantee Atomicity, writes performed by partially-executed transactions 
are \textit{undone} via an \textit{undo log}, restoring objects to
their last committed state. Unfortunately, this undo process
can leak information: the number of undo operations reveals the existence of ongoing
transactions, their length, and the number of operations
that they performed.

\subsection{Performance/functionality limitations}
\label{subsec:performance}

Current ORAMs align poorly with the need of modern OLTP workloads,
which must support large numbers of concurrent requests; in contrast,
most ORAMs admit little to no
concurrency~\cite{boyle16pram,williams12privatefs,stefanov2013oblivistore,sahin2016taostore}
(we benchmark the performance of sequential Ring ORAM in
Figure~\ref{graph:parallel}).

More problematically, ORAMs provide no support for
fault-tolerance. Adding support for Durability presents two main
challenges.  First, most ORAMs require the use of a
  \textit{stash} that temporarily buffers objects at the client and
  requires that these objects be written out to server storage in very
  specific ways (as we describe further in
  \S\ref{sec:background}). This process aligns poorly with
  guaranteeing Durability for transactions. Consider for example a
  transaction $T_1$ that reads the version of object $x$ written by
  $T_2$ and then writes object $y$. To recover the database to a
  consistent state, the update to $x$ should be flushed to cloud
  storage before the update to $y$. It may however not be possible to
  {\em securely} flush $x$ from the stash before $y$.  Second, ORAMs
  store metadata at the client to ensure that cloud storage observes a
  request pattern that is independent of past and currently executing
  operations. As we show in \S\ref{sec:durability}, recovering this
  metadata after a failure can lead to duplicate accesses that leak information.

\subsection{Introducing \sys{}}
These challenges motivate the need to co-design the
transactional and recovery logic with the underlying ORAM data
structure. The design should satisfy three goals: \one~security---the
system should not leak access patterns; \two~correctness---\sys{}
should guarantee that transactions are serializable; and
\three~performance---\sys{} should scale with the number of
clients. The principle of \textit{workload independence} underpins
\sys's security: the sequence of requests sent to cloud storage shoud
remain independent of the type, number, and access set of the
transactions being executed. Intuitively, we want \sys's sequence of
accesses to cloud storage to be statistically indistinguishable from a
sequence that can be generated by an \sys{} {\em simulator} with no
knowledge of the actual transactions being run by \sys{}. If this
condition holds, then observing \sys's accesses cannot reveal to the
adversary any information about \sys's workload. We formalize this
intuition in our security definition in \S\ref{sec:security}.

Much of \sys's novelty lies not in developing new concurrency control
or recovery mechanisms, but in identifying what standard database
techniques can be leveraged to lower the costs of ORAM while retaining
security, and what techniques instead subtly break obliviousness.

To preserve workload independence while
guaranteeing good performance in the presence of concurrent requests,
\sys{} centers its design around the notion of \textit{delayed
  visibility}. Delayed visibility leverages the observation that, on
the one hand, ACID consistency and Durability apply only when
transactions commit, and, on the other, commit operations can be
delayed. \sys{} leverages this flexibility to delay commit operations
until the end of \textit{fixed-size epochs}. This approach
allows \sys{} to \one~amortize the cost of accessing an ORAM over many
concurrently executing requests; \two~recover efficiently from
failures; and \three~preserve workload independence: the epochs'
deterministic structure allows \sys{} to decouple its externally
observable behavior from the specifics of the transactions being
executed.

%% file: background-current.tex
\section{Background}
\label{sec:background}

\textbf{O}blivious \textbf{R}emote \textbf{A}ccess \textbf{M}emory is
a cryptographic protocol that allows clients to access data outsourced
to an untrusted server without revealing what is being
accessed~\cite{goldreich1996software}; it generates a sequence of accesses to the server that is completely \textit{independent} of the operations issued by the
client.
We focus specifically
on \textit{tree-based} ORAMs, whose constructions are more efficiently implementable in real
systems: to date, they have been implemented in hardware~\cite{fletcher2015tinyoram,maas2013phantom} and as the basis
 for blockchain ledgers~\cite{cecchetti2017solidus} with reasonable overheads.
Most tree-based ORAMs  follow a 
similar structure: objects (usually key-value pairs) are mapped  to a
random leaf (or \textit{path}) in a binary tree and
physically reside (encrypted) in some tree node (or \textit{bucket}) along that path.
Objects are logically removed from the tree and
remapped to a new random path when accessed. These objects are eventually
flushed back to storage (according to their new path) as part of an \textit{eviction}
phase. Through careful scheduling, this write-back phase does not
reveal the new location of the objects; objects that cannot be flushed
are kept in a small client-side \textit{stash}.

\par \textbf{Ring ORAM}. \sys{} builds upon Ring ORAM~\cite{ren15ringoram}, a tree-based
ORAM with two appealing properties: a constant stash size and a fully
deterministic eviction phase.  \sys leverages these features for efficient
failure recovery.

As shown in Figure~\ref{fig:ringoramread}, server storage in Ring ORAM
consists of a binary tree of \textit{buckets}, each with a fixed
number $Z+S$ of \textit{slots}.  Of these, $Z$ are reserved for
storing actual encrypted data (\textit{real objects}); the remaining
$S$ exclusively store \textit{dummy objects}.  Dummy objects are
blocks of encrypted but meaningless data that appear indistinguishable
from real objects; their presence in each bucket prevent the server
from learning how many real objects the bucket contains and which slots contains them.  A random permutation (stored at the client) determines
the location of dummy slots.  In Figure~\ref{fig:ringoramread}, the
root bucket contains a real slot followed by
two dummy slots; the real slot contains the data object $a$; its left
child bucket instead contains dummy slots in positions one and
three, and an empty real slot in second position.

Client storage, on the other hand, is limited to \one~ a constant sized
\textit{stash}, which temporarily buffers objects that have yet to be
replaced into the tree and, unlike a simple cache, is essential
to Ring
ORAM's security guarantees; \two~the set of current
\textit{permutations}, which identify the role of each slot in each
bucket and record which slot have already been accessed (and 
marked {\em invalid)}; and \three~a \textit{position map}, which records the
random leaf (or \textit{path}) associated with every data object. In
Ring ORAM, objects are mapped to individual leaves of the tree but can
be placed in any one of the buckets along the path from the root to
that leaf. For instance, object $a$ in Figure~\ref{fig:ringoramread}
is mapped to \textit{path 4} but stored in the root bucket, while
object $b$ is mapped to \textit{path 2} and stored in the leaf bucket
of this path.
\begin{figure}[t]
\centering
\includegraphics[width=0.75\linewidth,valign=b]{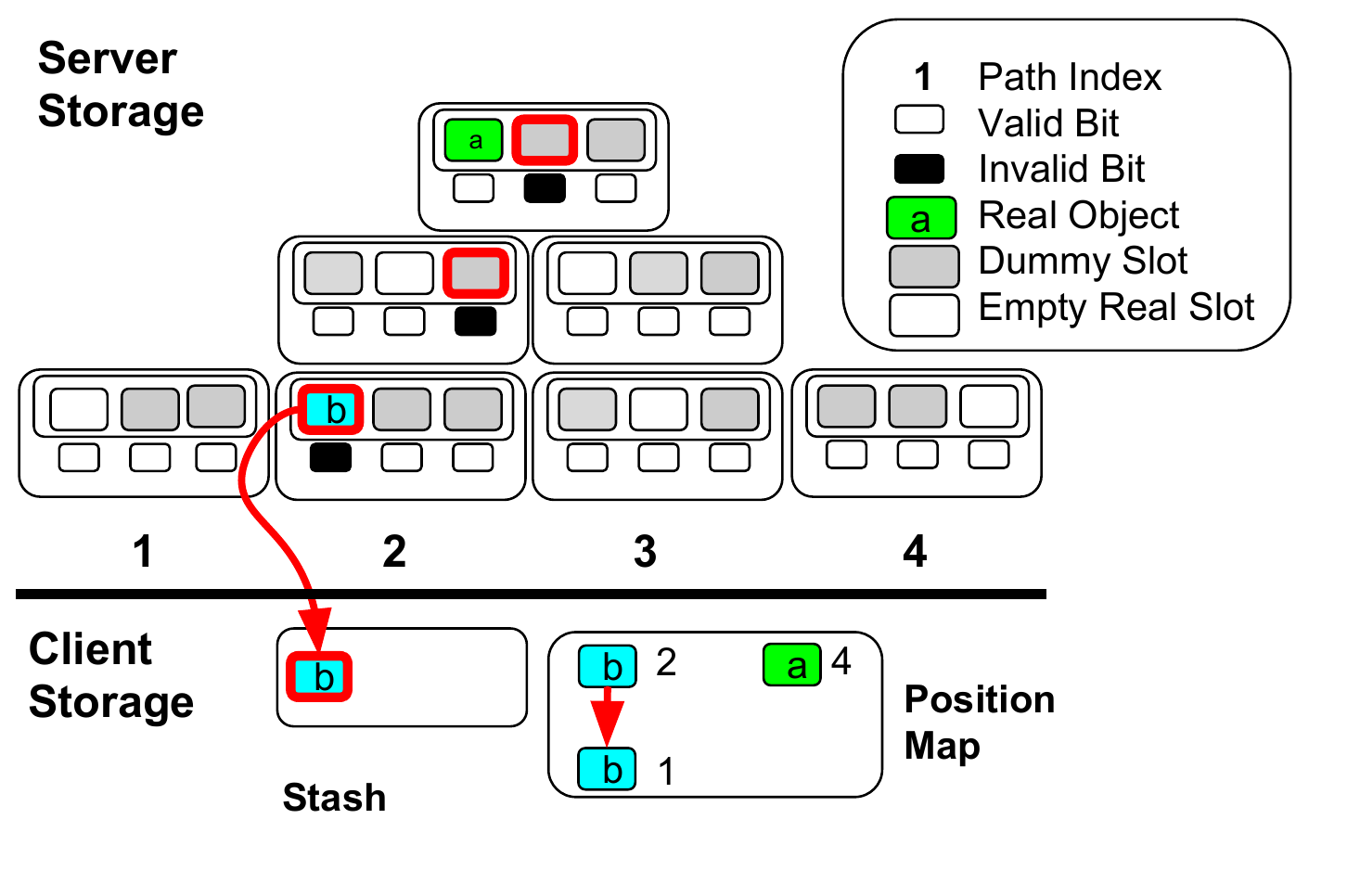}
    \caption{Ring ORAM - Read (Z=1, S=2)}%
\label{fig:ringoramread}
\end{figure}

Ring ORAM maintains two core invariants. First, each data object is
mapped to a new leaf chosen uniformly at random after every access, and is
stored either in the stash, or in a bucket on the path from the tree's
root to that leaf (\textbf{path invariant}).  Second, the physical
positions of the $Z+S$ dummy and real objects in each bucket are
randomly permuted with respect to all past and future writes to that
bucket (i.e., no slot can be accessed more than once between
permutations) (\textbf{bucket invariant}). The server never learns
whether the client accesses a real or a dummy object in the bucket, so
the exact position of the object along that path is never revealed.

Intuitively, the path invariant removes any correlation between two
accesses to the same object (each access will access independent
random paths), while the bucket invariant prevents the server from
learning when an object was last accessed (the server cannot
distinguish an access to a real slot from a dummy slot). Together,
these invariants ensure that, regardless of the data or type of
operation, all access patterns will look indistinguishable from a
random set of leaves and slots in buckets.

\par \textbf{Access Phase}. The procedures for read and write requests
is identical. To access an object $o$, the client first looks up $o$'s
path in the position map, and then reads one object from each bucket
along that path. It reads $o$ from the bucket in which it resides and
a valid dummy object from each other bucket, identified using its
local permutation map.  Finally, $o$ is remapped to a new path,
updated to a new value (if the request was a write), and added to the
stash; importantly, $o$ is not immediately written back out to cloud
storage.

Figure~\ref{fig:ringoramread} illustrates the steps involved in
reading an object $b$, initially mapped to path 2. The client reads a
dummy object from the first two buckets in the path (at slots two and
three respectively), and reads $b$ from the first slot of the bottom
bucket. The three slots accessed by the client are then marked as
invalid in their respective buckets, and $b$ is remapped to path 1.
To write a new  object $c$, the client would have to read three valid dummy
objects from a random path, place $c$ in the stash, and remap it to a
new path.

\par \textbf{Access Security}. Remapping objects to independent random
paths prevents the server from detecting repeated accesses to data,
while placing objects in the stash prevents the server from learning
the new path. Marking read slots as invalid forces every bucket access
to read from a distinct slot (each selected according to the random
permutation). The server consequently observes uniformly distributed
accesses (without repetition) independently of the contents of the
bucket. This lack of correlation, combined with the inability to
distinguish real slots from dummy slots, ensures that the server does
not learn if or when a real object is accessed.  Accessing dummy slots
from buckets not containing the target object (rather than real
slots), on the other hand, is necessary for efficiency: in combination
with Ring ORAM's  {\em eviction phase} (discussed next)  it lets the stash size remain constant  by preventing multiple real
objects from being addded to the stash on a single access.

\begin{figure}[t]
\centering
\includegraphics[width=0.75\linewidth,valign=b]{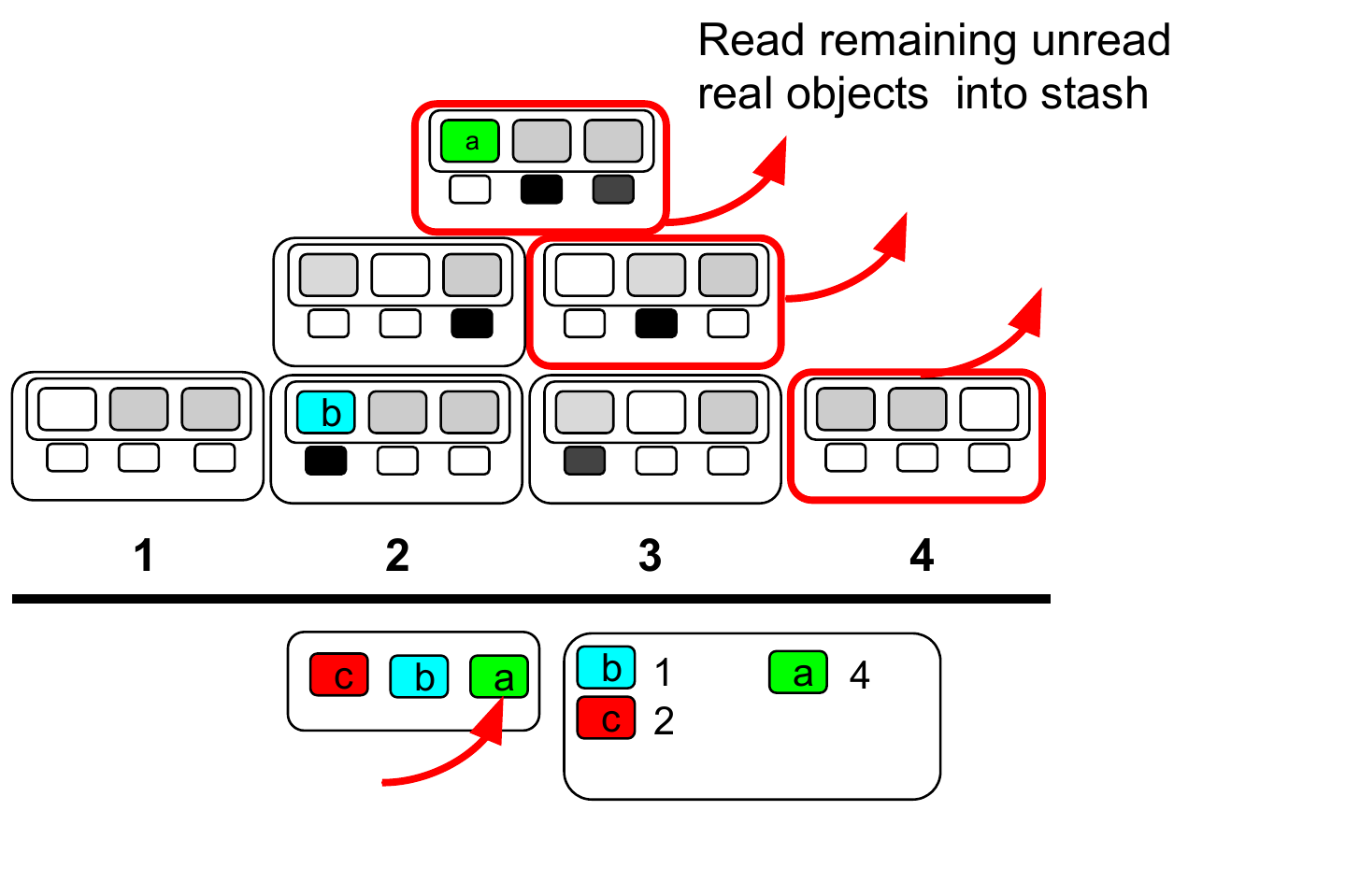}
\caption{Eviction - Read Phase}%
\label{fig:ringoramreadb}
\end{figure}

\par\textbf{Eviction Phase and Reshuffling}.
The aforementioned protocol falls short in two ways. First, if objects
are placed in the stash after each access, the stash will grow
unbounded. Second, all slots will eventually be marked as
invalid. Ring ORAM sidesteps these issues through two complementary
processes: \textit{eviction} and \textit{bucket reshuffling}. Every
$A$ accesses, the \emph{evict path} operation evicts objects
from the client stash to cloud storage.  It deterministically selects
a target path, flushes as much data as possible, and permutes each
bucket in the path, revalidating any invalid slots.  Evict path
consists of a read and write phase.  In the read phase, it retrieves
$Z$ objects from each bucket in the path: all remaining valid real objects, plus
enough valid dummies to reach a total of $Z$ objects read.  In the write
phase, it places each stashed object---including those read by the
read phase---to the deepest bucket on the target path that intersects
with the object's assigned path.  Evict path
  then permutes the real and dummy values in
each bucket along the target path, marking their slots as
\emph{valid}, and writes their contents to server storage.
Figure~\ref{fig:ringoramreadb} and~\ref{fig:ringoramwriteb} show the
evict path procedure applied to path 4.  In the read phase, evict path
reads the unread object $a$ from the root node and dummies from other
buckets on the path.  In the write phase
(Fig.~\ref{fig:ringoramwriteb}), $a$ is flushed to leaf 4, as its path
intersects completely with the target path.  Finally, we note that
randomness may cause a bucket to contain only invalid slots before its
path is evicted, rendering it effectively unaccessible. When this
happens, Ring ORAM restores access to the bucket by performing an
\textit{early reshuffle} operation that executes the read phase and
write phase of evict path only for the target bucket.

\par\textbf{Eviction Security}.
The read phase leaks no information about the contents of a given bucket. It systematically reads exactly $Z$ valid objects from the bucket,
selecting the valid real objects from the $z$ real objects in the bucket,
padding the remaining $Z-z$ required reads with a random subset of the $S$
dummy blocks. The random permutation and randomized encryption ensure that the server learns no information about how many real objects exist, and how many have been accessed. Similarly, the write phase hides the values and locations of objects written. At every bucket, the storage server observes only a newly encrypted and permuted set of objects, eliminating any correlation between past and future accesses to that bucket. Together, the read and write phases ensure that no slot is accessed
more than once between reshuffles, guaranteeing the bucket invariant.

\begin{figure}[t]
\centering
\includegraphics[width=0.75\linewidth,valign=b]{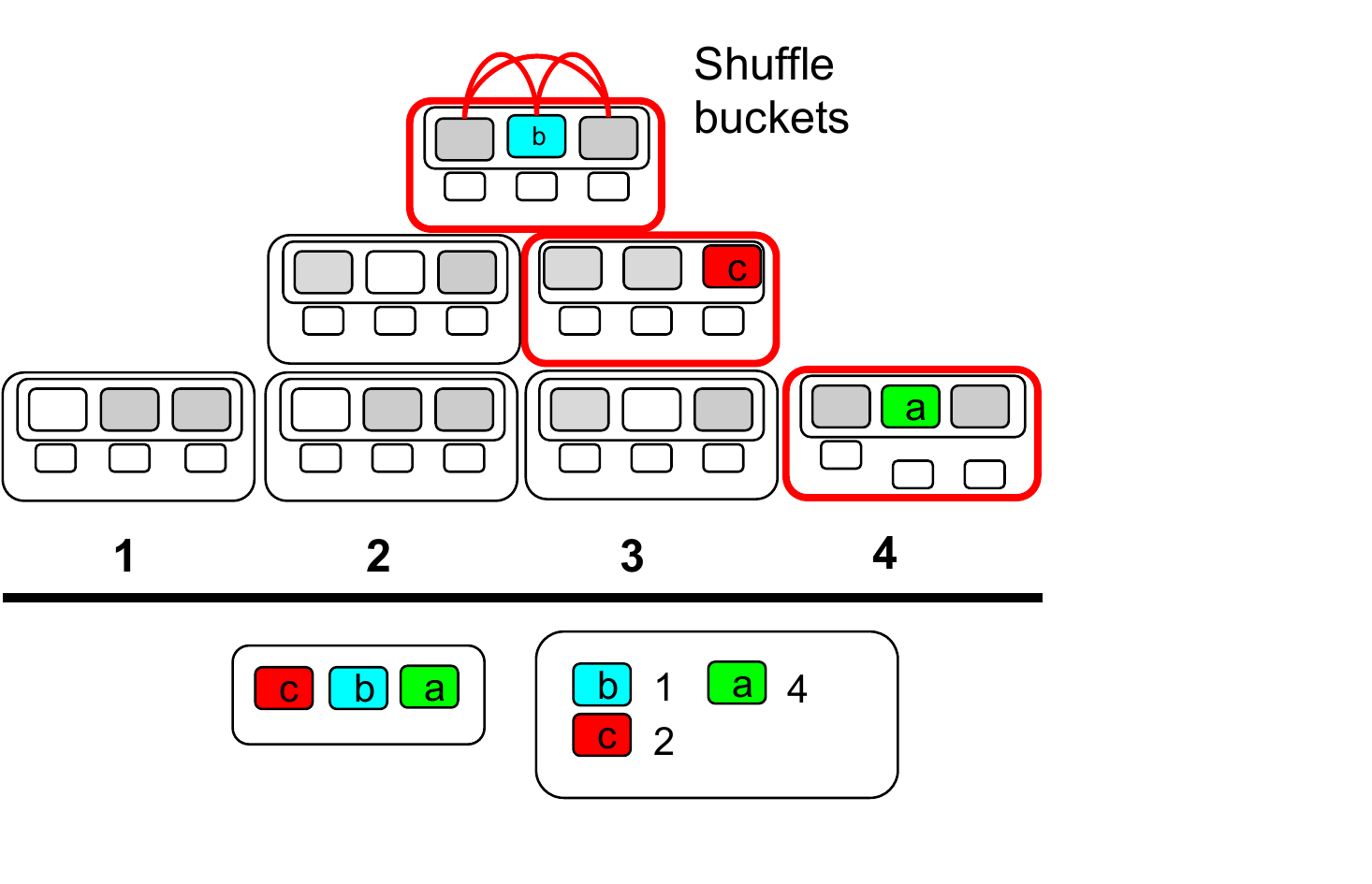}
\caption{Eviction - Write Phase}%
\label{fig:ringoramwriteb}
\end{figure}

Similarly, the eviction process leaks no information about the paths of the newly evicted objects: since all paths intersect at the root and the server cannot infer the contents of any individual bucket, any object in the stash may be flushed during \emph{any} evict path.
Moreover, since all paths intersect at the root, any object in the stash may be
flushed during \emph{any} evict path.

%% file: design-current.tex
\section{System Architecture}
\label{sec:architecture}

\sys, like most privacy-preserving
systems~\cite{stefanov13oblivistore,williams12privatefs,sahin2016taostore}
consists of a centralized trusted component, the \emph{proxy},
that communicates with a fault-tolerant but untrusted entity,
\emph{cloud storage} (Figure~\ref{fig:sysarch}). The proxy handles concurrency control, while the
untrusted cloud storage stores the private data. \sys{} ensures that
requests made by the proxy to the cloud storage over the untrusted
network do not leak information.  We assume that the proxy
can crash and that when it does so, its state is lost. This two-tier design allows
applications to run a lightweight proxy locally and
delegate the complexity of fault-tolerance to cloud storage.

The proxy has two components: \one~a \emph{concurrency control unit}
and \two~a \emph{data manager} comprised of a \emph{batch manager} and
an \emph{ORAM executor}. The batch manager periodically schedules
fixed-size batches of client operations that the ORAM executor then
executes on a parallel version of Ring ORAM's algorithm.  The executor
accesses one of two units located on server storage: \emph{the ORAM
  tree}, which stores the actual data blocks of the ORAM; and
\emph{the recovery unit}, which logs all non-deterministic accesses to
the ORAM to a write-ahead log~\cite{mohan1992aries} to enable secure
failure recovery (\S\ref{sec:durability}).

%% file: proxy-design-current.tex
\section{Proxy Design}
\label{sec:proxy}

The proxy in \sys{} has three goals: guarantee good performance, preserve
correctness, and guarantee security.
To meet these goals, \sys{} designs the proxy around the concept of
epochs.  The proxy partitions time into a set of fixed-length,
non-overlapping epochs. Epochs are the granularity at which \sys{}
guarantees durability and consistency. Each transaction, upon arriving
at the proxy, is assigned to an epoch and clients are notified of
whether a transaction has committed only when the epoch ends. Until
then, \sys{} buffers all updates at the proxy.

 This flexibility boosts \textit{performance} in
two ways. First, it allows \sys{} to implement a multiversioned
concurrency control (MVCC) algorithm on top of a single versioned
Ring ORAM. MVCC algorithms can significantly improve throughput by
allowing read operations to proceed with limited blocking. These
performance gains are especially significant in the presence of long-running
transactions or high storage access latency, as is
often the case for cloud storage systems. Second, it
reduces traffic to the ORAM, as only the database state at the
end of the epoch needs to be written out to cloud storage.

\begin{figure}
    \centering
\includegraphics[width=\linewidth]{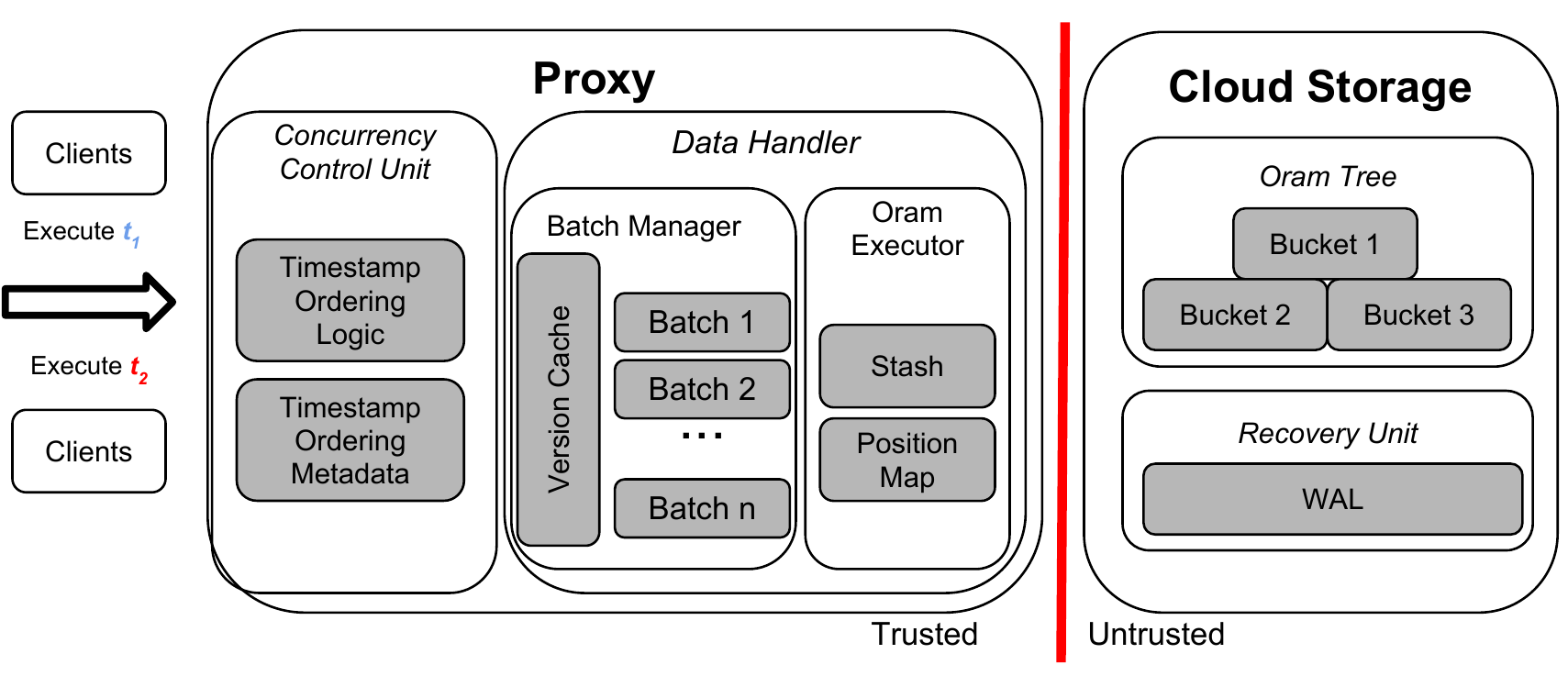}
    \caption{System Architecture}
    \label{fig:sysarch}
\end{figure}

Importantly,  \sys's choice to enforce consistency and
durability only at epoch boundaries does not affect \textit{correctness};
transactions continue to observe a serializable and recoverable
schedule (i.e., committed transactions do not see  writes from 
aborted transactions). 

For transactions executing concurrently within the same epoch,
serializability is guaranteed by concurrency control;
transactions from different epochs are naturally serialized by the
order in which the proxy executes their epochs. No transaction can
span multiple epochs; unfinished transactions at epoch boundaries are
aborted, so that no transaction is ongoing during epoch
changes. 

Durability is instead achieved by enforcing epoch
fate-sharing~\cite{tu2013silo} during proxy or client crashes: \sys{}
  guarantees that either all \textit{completed} transactions
  (i.e., transactions for which a commit request
  has been received) in the
  epoch are made durable or all transactions abort. This way, no
  \textit{committed} transaction can ever observe
  non-durable writes.

  Finally, the deterministic pattern of execution that epochs impose
  drastistically simplifies the task of guaranteeing workload
  independence: as we describe further below, the frequency and timing
  at which requests are sent to untrusted storage are fixed and
  consequently independent of the workload.

  The proxy processes epochs with two modules: the concurrency
  control unit (CCU) ensures that execution remains serializable, while 
  the data handler (DH) accesses the actual data objects. We
  describe each in turn.

\subsection{Concurrency Control}
\label{subsec:proxy:cc}

\sys{}, like many existing commercial databases~\cite{innodb,postgres},
uses multiversioned concurrency control~\cite{bernstein1983mcc}.
\sys specifically chooses multiversioned timestamp ordering
(MVTSO)~\cite{bernstein1983mcc,reed1983atomic} because it allows
uncommitted writes to be immediately visible to concurrently executing
transactions. To ensure serializability,  transactions log the set of transactions
whose uncommitted values they have observed (their write-read
dependencies) and  abort if any of their dependencies fail to
commit. This optimistic approach is critical to \sys's performance: it allows transactions
within the same epoch to see each other's effects even as
\sys delays commits until the epoch ends.  In contrast, a pessimistic
protocol like two-phase locking~\cite{eswaran1976ncp}, which precludes
transactions from observing uncommitted writes, would artificially increase
contention by holding exclusive write-locks for the duration of an
epoch.   When a transaction starts, MVTSO assigns it  a
unique timestamp that determines its serialization order.  A write
operation creates a new object version marked with its transaction's
timestamp and inserts it in the \textit{version chain} associated with
that object. A read operation returns the object's latest version with a
timestamp smaller than its transaction's timestamp. Read operations
further update a \textit{read marker} on the object's version chain with their
transaction's timestamp. Any write operation with a smaller timestamp
that subsequently tries to write to this object is aborted, ensuring
that no read operation ever fails to observe a write from a transaction that should
have preceded it in the serialization order.

Consider for example the set of transactions executing in Figure~\ref{fig:batchinglogic}.
Transaction $t_1$'s update to object $a$ ($w(a_1)$) is
immediately observed by transaction $t_3$ ($r_3(a_1)$). $t_3$ becomes 
dependent on $t_1$ and can only commit once $t_1$ also
commits. In contrast, $t_2$'s write to object $d$ causes
$t_2$ to abort: a transaction with a higher timestamp ($t_3$)
had already read version $d_0$, setting the version's read marker to $3$.

\begin{figure}
\centering
\includegraphics[width=\linewidth]{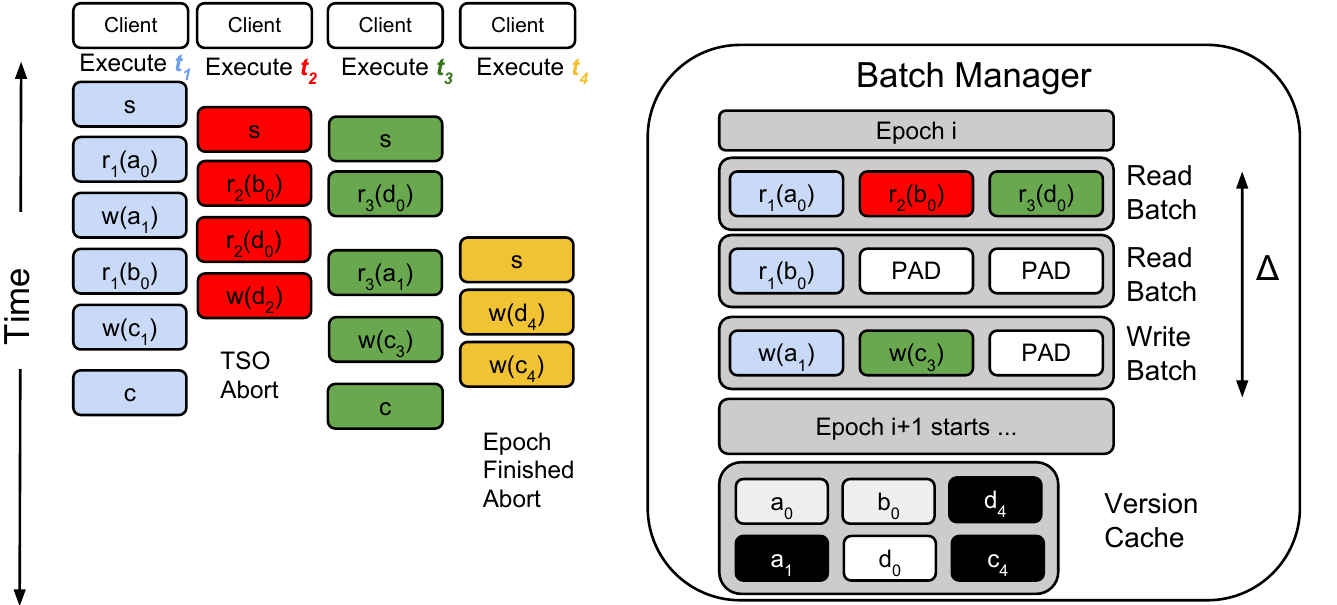}
    \caption{Batching Logic - $r_x(a_y)$ denotes
    that transaction $t_x$ reads the version of object
    $a$ written by transaction $t_y$}
\label{fig:batchinglogic}
\end{figure}

\subsection{Data Handler}
\label{subsec:proxy:dh}

Once a version is selected for reading or writing, the DH becomes
responsible for accessing or modifying the actual object.  Whereas it
suffices to guarantee durability and consistency only at epoch boundaries,
security must hold at all times, posing two key challenges. First, the
number of requests executed in parallel can leak information, e.g.,
data dependencies within the same
transaction~\cite{bindschaedler2015curious,sahin2016taostore}.
Second, transactions may abort (\S\ref{subsec:proxy:cc}), requiring
their effects to be rolled back without revealing the existence of
contended objects~\cite{arasu2013cipherbase,sheff2016sss}.  To
decouple the demands of these workloads from the timing and set of
requests that it forwards to cloud storage, \sys{} leverages the
following observation: transactions can always be re-organized so that all
reads from cloud storage execute before all
writes~\cite{kung1981occ,zhang2015tapir,mehdi17occult,corbett2013spanner}.
Indeed, while operations within a transaction may depend on the data
returned by a read from cloud storage, no operation
depends on the execution of a
write.
Accordingly, \sys{} organizes the DH into a read phase and a write
phase: it first reads all necessary objects from
cloud storage, before applying all writes.

\par \textbf{Read Phase}. \sys{} splits each epoch's read phase into a
\textit{fixed} set of $R$ \textit{fixed-sized} read batches ($b_{read}$) that are
forwarded to the ORAM executor at \textit{fixed}
intervals ($\Delta_{epoch}$). This deterministic structure allows
\sys{} to execute dependent read operations without revealing the
internal control flow of the epoch's transactions. 
Read operations are assigned to the epoch's next unfilled read
batch. If no such batch exists, the transaction is
aborted.
Conversely, before a batch is forwarded to the ORAM executor,
all remaining empty slots are padded with dummy requests.
\sys{} further \textit{deduplicates} read operations that access the
same key. As we describe in \S\ref{sec:parallel}, this step is necessary for security
since parallelized batches may leak information unless
requests all access distinct
keys~\cite{boyle16pram,williams12privatefs}. Deduplicating requests
also benefits performance by increasing the number of operations
that can be served within a fixed-size batch.

\par \textbf{Write Phase}. While transactions execute, \sys{}
buffers their write operations into a \textit{version cache}
that maintains all object versions created by transactions in the epoch.
At the end of an epoch, transactions that have yet to finish executing (recall that epochs terminate
at fixed intervals) are aborted and their operations are removed. The latest versions
of each object in the version cache according to the version chain are then aggregated in a fixed-size \textit{write
batch} ($b_{write}$) that is forwarded to the ORAM executor, with additional padding if
necessary. 

This entire process, including write buffering and deduplication, must
not violate serializability.  The DH guarantees that write buffering respects
serializability by directly serving reads from the
  version cache for objects modified in the current
  epoch. It guarantees
  serializability in the presence of duplicate requests by 
  only including
   the last write of the version chain in a write batch.
   Since \sys's epoch-based design guarantees that transactions from a later
  epoch are serialized after all transactions from an earlier epoch,
  intermediate object versions can be safely discarded.
  In
  this context, MVTSO's requirement that transactions observe the
  \textit{latest} committed write in the serialization order reduces 
  to transactions reading the tail of the previous epoch's version
  chain. 

In the presence of failures, \sys{} guarantees serializability
and recoverability by enforcing epoch fate sharing: either all
transactions in an epoch are made durable or none are. If a failure
arises during epoch $e_i$, the system simply recovers to epoch
$e_{i-1}$, aborting all transactions in epoch $e_i$. Once again, this
flexibility arises from \sys{} delaying commit notifications until epoch
boundaries.

\par \textbf{Example Execution}. We illustrate the batching logic once
again with the help of Figure~\ref{fig:batchinglogic}. Transactions
$t_1$,$t_2$,$t_3$ first execute read operations. These operations are
aggregated into the first read batch of epoch $i$. The values returned
by these reads are then \textit{cached} into the version cache.  $t_2$
then executes a write operation, which \sys{} also buffers into the
version cache. When executing $r_2(d_0)$), $t_3$ reads object $d$ directly
from the version cache (we discuss the security of this step in
the next section). Similarly, $r_1(a_1)$ reads the buffered
uncommitted version of $a$. In contrast, \sys{} schedules $r_1(b_0)$
to execute as part of the next read batch as $b_0$ is not present in
the version cache. The read batch is then padded to its fixed
$b_{read}$ size and executed.  $t_4$ contains no read operations: its
write operations are simply executed and buffered at the version
cache.  \sys{} then finalizes the epoch by aborting all transactions
(and their dependencies) that have not yet finished executing: $t_4$
is consequently aborted. Finally, \sys{} aggregates the last version
of every update into the write batch (skipping version $c_1$ of object
$c$ for instance, instead only writing $c_2$), before notifying
clients of the commit decision.

\subsection{Reducing Work}

\sys{} reduces work in two additional ways: it caches reads within an
epoch and allows Ring ORAM to execute write operations without also
executing dummy queries. While these optimizations may appear
straightforward, ensuring that they maintain workload independence requires care. 

\par \textbf{Caching Reads}. Ring ORAM maintains a client-side stash
(\S\ref{sec:background}) that stores ORAM blocks until their eviction
to cloud storage. Importantly, a request for a block present in the stash
still triggers a dummy request: a dummy object is
still retrieved from each bucket along its path. While this
access may appear redundant at first, it is in fact
necessary to preserve \textit{workload independence}: removing
it removes the guarantee that the set of paths that \sys{} requests
from cloud storage is uniformly distributed. In particular, blocks present in the stash are more
likely to be mapped to paths farther away from
the one visited by the last evict path, as 
they correspond to paths that could not be flushed: buckets have limited
space for real blocks and blocks mapped to paths that only intersect near the
top of the tree are less likely to find a free slot to which
they can be flushed.
The degree to which
this effect skews the distribution leaks information about the stash
size, and, consequently, about the workload. To illustrate, consider the execution
in Figure~\ref{fig:caching}. Objects mapped to paths 1 and 2 ($a$, $b$, and $f$) were not
flushed from the stash in the previous eviction of path 4. When these objects
are subsequently accessed, naively reading them from the stash without performing
dummy reads skews the set of paths accessed toward the right subtree (paths 3 and 4)

\sys{} securely mitigates some of this work by drawing a novel
distinction between objects that are in the stash as a result of a logical access
and those present because they could not be evicted.
The former can be safely accessed without performing
a dummy read, while the latter cannot.
Objects present in the stash following a logical access are mapped to independently
uniformly distributed paths. Ring ORAM's
path invariant ensures that, without caching, 
the set of accessed paths is uniformly distributed. Removing an
independent uniform subset of those paths (namely, the dummy requests)
will consequently not change the distribution. Thus, 
caching these objects, and filling out a read batch with other real or dummy requests, 
preserves the uniform distribution of paths and leaks no information.
\sys consequently allows all read objects to be placed in the version cache for the duration of the epoch.
Objects $a$, $b$, $d$ are, for instance, placed in the version cache in Figure~\ref{fig:batchinglogic},
allowing read $r_2(d_0)$ to read $d$ directly from the cache.
In contrast, objects
present in the stash because they could not be evicted are mapped
to paths that skew away from the latest evict path. Caching these objects would
consequently skew the distribution of requests sent to the storage away from 
a uniform distribution, as illustrated in Figure~\ref{fig:caching}.

\begin{figure}[!t]
\centering
\includegraphics[width=\linewidth,valign=b]{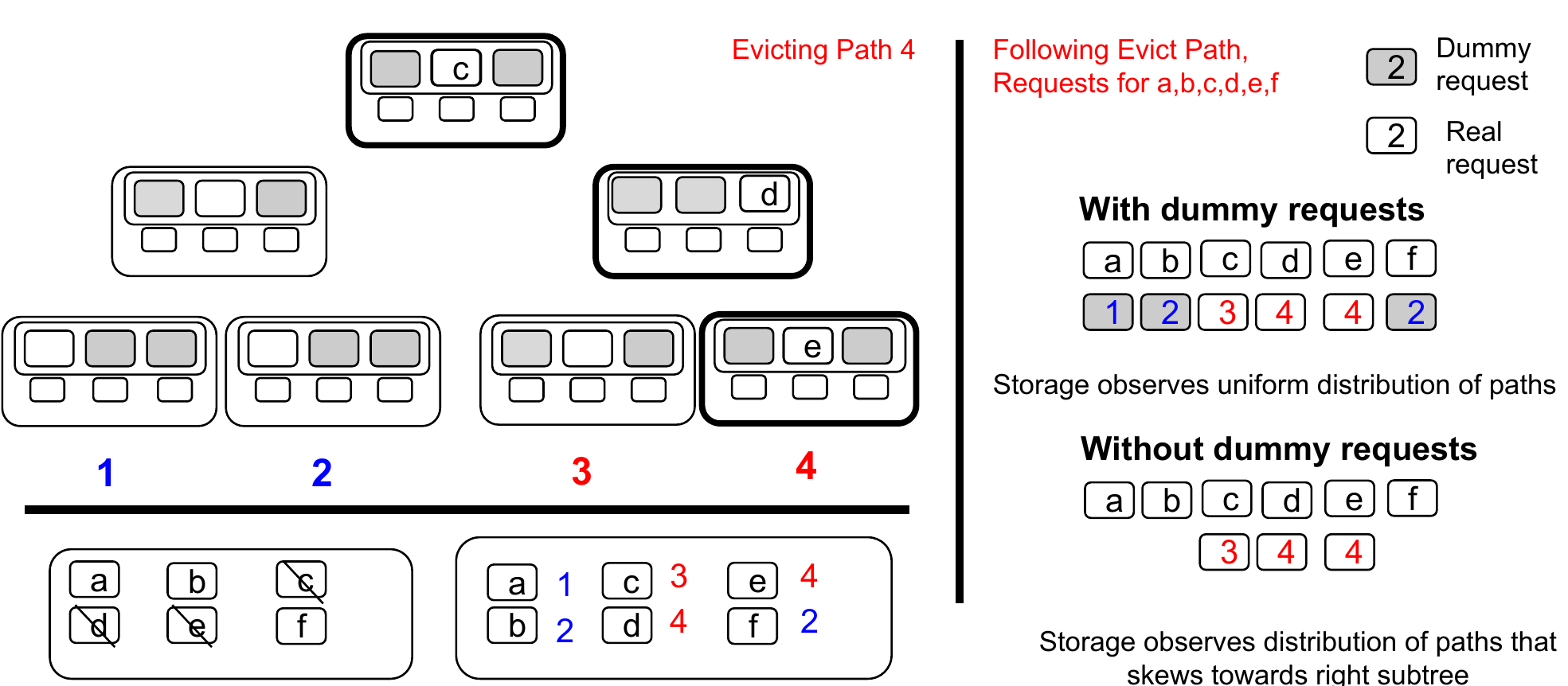}
    \caption{Skew introduced by caching arbitrary objects}%
\label{fig:caching}
\end{figure}

\par\textbf{Dummiless Writes}.
Ring ORAM must hide whether requests correspond to read or write
operations, as the specific pattern in which these operations are
interleaved can leak information~\cite{zheng17opaque}; that is why
Ring ORAM executes a read operation on the ORAM for every access. In
contrast, since transactions can always perform all reads before all
writes, no information is leaked by informing the storage server that
each epoch consists of a fixed-size sequence of 
potentially dummy reads followed by a fixed-size
sequence of potentially dummy writes.
\sys thus modifies Ring ORAM's algorithm to directly
place the new version of an object in the stash, without executing the
corresponding read. Note, though, that \sys continues to increment the evict
path count on write operations, a necessary step  to preserve
the bounds on the stash size, which is important for durability (\S\ref{sec:durability}).

\subsection{Configuring \sys{}}
\label{subsec:configure}

\sys{}'s good performance hinges on appropriately configuring
the size/frequency of batches and ORAM tree for a target
application. Table~\ref{table:config} summarizes the parameter space.

\par \textbf{Ring ORAM}. Configuring Ring ORAM first requires choosing an appropriate $Z$ parameter. Larger 
values of $Z$ reduce the total size of the ORAM on cloud-storage by decreasing the required height
of the ORAM tree and decrease eviction 
frequency (reducing network/CPU overhead). In contrast, this increase the maximum stash size.
Traditional ORAMs thus choose the largest value of Z for which the stash
size fits on the proxy. \sys{} adds an additional consideration: for durability 
(as we describe in \S\ref{sec:durability}), the stash must be synchronously written out 
every epoch. One must thus take into account the throughput loss associated
with the stash writeback time.  Given an appropriate value of Z, \sys{} then chooses
L, S, and A according to the analytical model proposed in \cite{ren15ringoram}.

\begin{table}[t]
    \centering
\begin{tabular}{|c|c|}
\hline
   $N$ & Number of real objects \\
\hline
$Z$ & Number of real slots \\
\hline
$S$ & Number of dummy slots \\
\hline
$A$ & Frequency of evict path \\
\hline
$L$ & Number of levels in the ORAM tree \\
\hline
$R$ & Number of read batches \\
\hline
$b_{read}$ & Size of a read batch \\
\hline
$b_{write}$ & Size of a write batch \\
\hline
$\Delta$ & Batch frequency \\
\hline
\end{tabular}
\caption{\sys{}'s configuration parameters}
\label{table:config}
\end{table}

\par \textbf{Epochs and batching}. Identifiying the appropriate size and number of batches hinges on several
considerations. First,  \sys{} must provision sufficiently many read batches ($R$) to handle control flow dependencies within transactions. A transaction that executes in sequence five dependent read operations, will for instance
require five read batches to execute (it will otherwise repeatedly abort).
Second, the ratio of reads ($R * b_{read}$) to
writes ($w_{write}$) must closely approximate the application's read/write ratio.
An overly large write batch will waste resources as it will be padded with
many dummy requests. A write batch that is too small will lead to frequent aborts caused by the batch filling up.  Third,
the size of a read or write batch (respectively
$b_{read}$ and $b_{write}$) defines the
degree of parallelism that can be extracted. The desired batch size is
thus a function of the concurrent load of the system, but also of
hardware considerations, as increasing
parallelism beyond an I/O or CPU bottleneck serves no purpose. Finally, the number and frequency of read batches within an epoch
increases overall latency, but reduces amortized resource costs through
caching and operation pipelining (introduced in
\S\ref{sec:parallel}). Latency-sensitive applications may favor smaller
batch sizes, while others may prefer longer epochs, but lower overheads.

\par \textbf{Security Considerations}. \sys{} does not attempt to hide the size and frequency of batches from the storage
server (we formalize this leakage in \S\ref{sec:security}). Carefully tuning the size and frequency of batches to best match a given application
may thus leak information about the application itself. An OLTP application, for instance,
will likely have larger batch sizes ($b_{read}$), but fewer read batches ($R$), as OLTP applications
sustain a high concurrent load of fairly short transactions. OLAP applications will
prefer small or non-existent write batches ($b_{write}$), as they are predominantly
read-only, but require many read batches to support the complex joins/aggregates
that they implement. 
\sys{} does not attempt to hide the type of application that is being run.
It does, however, continue to hide what data is being accessed and what transactions are
currently being run at any given point in time. While \sys's configuration parameters may,
for instance, suggest
that a medical application like FreeHealth is being run, they do not in any way leak information
about how, when, or which patient records are being accessed.

%% file: oram-design-current.tex
\section{Parallelizing the ORAM}
\label{sec:parallel}

Existing ORAM constructions make limited use of parallelism. Some
allow requests to execute concurrently between eviction or shuffle
phases~\cite{boyle16pram,williams12privatefs,sahin2016taostore}, while
others target intra-request parallelism to speed up execution of a
single request~\cite{lorch2013shroud}.  \sys{} explicitly targets both
forms of parallelism. Parallelizing Ring ORAM presents three
challenges: \one preserving the correct abstraction of a sequential
datastore, \two enforcing security by concealing the position of real
blocks in the ORAM (thereby maintaining workload independence), and
\three preserving existing bounds on the stash size. While these
issues also arise in prior work~\cite{sahin2016taostore}, the
idiosyncrasies of Ring ORAM add new dimensions to these challenges.

\par \textbf{Correctness}.
\sys{} makes
two observations. First, while all operations 
conflict at the Ring ORAM tree's root, they 
can be split into  suboperations that
access mostly disjoint \textit{buckets} (\S\ref{sec:background}). Second, 
conflicting bucket operations can be further parallelized by
distinguishing accesses to the bucket's metadata from those to 
its physical data blocks.

\sys{} draws from the theory of multilevel
serializability~\cite{weikum1991multilevel}, which guarantees that an
execution is serializable if the system enforces level-by-level
serializability: if operation $o$ is ordered before $o'$ at level $i$,
all suboperations of $o$ must precede conflicting suboperations of
$o'$. Thus, if \sys{} orders conflicting operations at a level $i$, it
enforces the same order at level $i+1$ for all their conflicting
suboperations; conversely, if two operations do not conflict at level
$i$, \sys{} executes their suboperations in parallel.  To this end,
\sys{} simply tracks dependencies across operations and orders
conflicting suboperations accordingly.  \sys{}
extracts further parallelism in two ways. First, since in Ring ORAM
reads to the same bucket between consecutive
eviction or reshuffling operations always target different physical
data blocks (even when bucket operations conflict on metadata access),
\sys{} executes them in parallel. Second, \sys's own batching logic
ensures that requests within a batch touch different objects,
preventing read and write methods from ever conflicting.  Together,
these techniques allow \sys{} to execute most requests and evictions
in parallel.

We illustrate the dependency
tracking logic in Figure~\ref{fig:multilevel}. The read operation
to path 1 conflicts with the evict path for path 2, but only at the root (bucket 1).
Thus, reads to buckets 2 and 3 can proceed concurrently, even though accesses to the root's metadata must be
serialized, as both operations update the bucket access counter
and valid/invalid map (\S\ref{sec:background}).

\begin{figure}%
\centering
    \includegraphics[width=\linewidth]{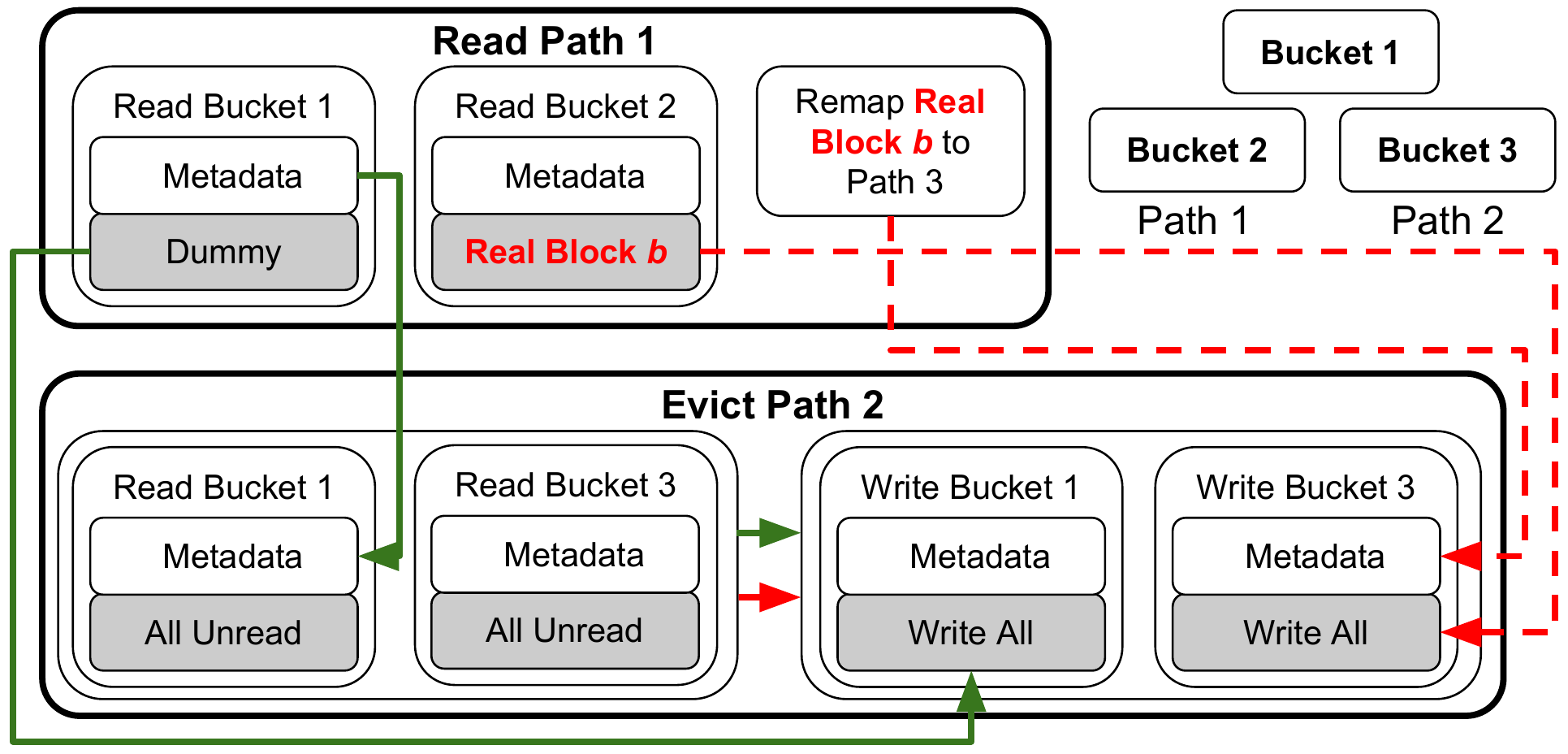}
    \caption{Multilevel Pipelining for a read of path 1 and an evict path of path 2
    executing in parallel. Solid green lines represent physical
    dependencies and dashed red lines represent data dependencies.
    Inner blocks represent nested operations}
    \label{fig:multilevel}
\end{figure}

\par \textbf{Security}.  For security, \sys's parallel evict path
operation must flush the same blocks flushed by a sequential
implementation. Reproducing this behavior without sacrificing
parallelism is challenging.  It requires that all real objects brought
in during the last $A$ accesses be present in the stash when data is
flushed, which may introduce \emph{data dependencies}. Unlike
dependencies that arise between operations that access the same
physical location in cloud storage, these dependencies are not
a deterministic function of an epoch's operations already known to the
adversary.

Consider, for instance, block $b$ in Figure~\ref{fig:multilevel}.  In
a sequential implementation, $b$ would enter the stash as a result of
reading path 1 and be flushed to bucket 3 by the following evict path.
Thus, evict path would have to \emph{wait} until $b$ is placed in the
stash.  Honoring these dependencies opens a timing channel:  delay
in flushing certain blocks can reveal object 
placement.  
As blocks holding real objects can exist anywhere in the tree and be
remapped to any path,  it follows that it  is never secure to execute
an eviction operation until all previous access operations have
terminated.

\sys mitigates this restriction by again leveraging delayed visibility
and the idea to separate read and write operations within an
epoch---but with an important difference. In \S\ref{subsec:proxy:dh}
the proxy created separate batches for {\em logical} read and write
operations; to improve parallelism, \sys{}, expanding on an idea used
by Shroud~\cite{lorch2013shroud}, assigns to separate phases within an
epoch the {\em physical} read and write operations that underlie each
of those logical operations. The read phase computes all necessary metadata and executes the set of
physical read operations for all logical read path, early reshuffle,
and evict path operations. This set is workload independent, so its
operations need not be delayed.  Physical writes, however, are only
flushed at the end of an epoch.  The proxy can again apply write
deduplication: if a bucket is
repeatedly modified during an epoch, only the last version must be
written back.  Reads that should have read an intermediate write are
served locally from the buffered buckets.

The adversary thus always observes a set of reads to random paths
followed by a deterministic set of writes independent of the
contents of the ORAM and, consequently, of the workload. Data
dependencies between read and evict operations no longer create a
timing channel. Meanwhile parallelism remains high, as the physical
blocks accessed in each phase are guaranteed to be distinct---Ring
ORAM directly guarantees this for reads, while bucket deduplication
does it for writes.

%% file: durability-design-current.tex
 \section{Durability}
\label{sec:durability}

\sys{} guarantees durability at the granularity of epochs: after a
crash, it recovers to the state of the last failure-free epoch. \sys
adds two demands to the need of recovering to a consistent state:
recovery should leak no information about past or future transactions,
and it should be efficient, accessing minimal data from cloud storage.
\sys{} guarantees the former by ensuring that recovery logic and data
logged for recovery maintain workload independence
(\S\ref{sec:threat}). It strives towards the latter by leveraging the
determinism of Ring ORAM.

\par \textbf{Consistency}. \sys{} recovery logic relies on two
well-known techniques: write-ahead logging~\cite{mohan1992aries} and
shadow paging~\cite{gray1981shadow}.  \sys{} mandates that
transactions be durable only at the end of an epoch; thus, on a proxy
failure, all ongoing transactions can be aborted, and the system
reverted to the previous epoch.  To make this possible, \sys{} must
~\one recover the proxy metadata lost during the proxy crash,
and ~\two ensure that the ORAM does not contain any of the aborted
transactions' updates.  To recover the metadata, \sys{} logs three
data structures before declaring the epoch committed: the position
map, the permutation map, and the stash.  The position map and the
permutation map identify the position of real objects in the ORAM tree
(respectively, in a path and in a bucket); logging them prevents the
recovery logic from having to scan the full ORAM to recover the
position of buckets. Logging the stash is necessary for
correctness.  As eviction may be unable to flush the entire stash,
some newly written buckets may be present only in the stash, even at
epoch boundaries. Failing to log the stash could thus lead to data
loss.

To undo partially executed transactions, \sys{} adapts the traditional
copy-on-write technique of shadow paging~\cite{gray1981shadow}: rather
than updating buckets in place, it creates new versions of each bucket
on every write. \sys{} then leverages the inherent determinism of Ring
ORAM to reconstruct a consistent snapshot of the ORAM at a given
epoch. In Ring ORAM, the current version of a bucket (i.e. the number
of times a bucket has been written) is a deterministic function of the
number of prior evict paths.  The number of evict paths per epoch is
similarly fixed (evict paths happen every $A$ accesses, and epochs are
of fixed size). \sys{} can then trivially revert the ORAM on failures
by setting the evict path counter to its value at the end of the last
committed epoch. This counter determines the number of evict paths
that have occurred, and consequently the object versions of the
corresponding epoch.

\par \textbf{Security}. \sys{} ensures that \one~the information
logged for durability remains independent of data accesses, and
\two~that the interactions between the failed epoch, the recovery
logic, and the next epoch preserve workload independence.

\sys{} addresses the first issue by encrypting the position map and
the contents of the permutations table. It similarly encrypts the
stash, but also \textit{pads} it to its maximum size, as determined in
canonical Ring ORAM~\cite{ren15ringoram}, to prevent it from
indicating skew (if a small number of objects are accessed frequently,
the stash will tend to be smaller).

The second concern requires more care: workload independence must hold
before, during, and after failures. Ring ORAM guarantees workload
independence through two invariants: the bucket invariant and the path
invariant (\S\ref{sec:background}). Preserving bucket slots from
being read twice between evictions is straightforward. \sys{} simply
logs the invalid/valid map to track which slots have already been read
and recovers it during recovery; there is no need for encryption, as
the set of slots read is public information.
Ensuring that the ORAM continues to observe a uniformly distributed set of paths is instead
more challenging. Specifically, read requests from partially executed transactions
can potentially leak information, even when recovering to the previous
epoch.  Traditionally, databases simply \textit{undo} 
partially executed transactions, mark them as aborted, and proceed as
if they had never existed. From a security standpoint, however, these
transactions were still observed by the adversary, and thus may leak
information.  Consider a transaction accessing object $a$ (mapped to
path $1$) that aborts because of a proxy failure. Upon recovery, it is
likely that a client will attempt to access $a$ again.  As the
recovery logic restores the position map of the previous epoch, that
new operation on $a$ will result in another access to path $1$,
revealing that the initial access to path $1$
was likely real (rather than padded), as the probability of collisions
between two uniformly chosen paths is  low.  To mitigate this
concern while allowing clients to request the same objects after
failure, \sys{} durably logs the list of paths and slot indices that
it accesses, before executing the actual requests, and replays
those paths during recovery (remapping any real blocks). While this
process is similar to traditional database redo
logging~\cite{mohan1992aries}, the goal is different. \sys{} does not
try to reapply transactions (they have all aborted), but instead
forces the recovery logic to be deterministic: the adversary 
always sees the paths from the aborted epoch repeated after a failure.

\par \textbf{Optimizations}. To minimize the overhead of checkpointing,
\sys{} checkpoints deltas of the position, permutation, and valid/invalid
map, and only periodically checkpoints the full data structures. 
While
the number of changes to the permutation and valid/invalid maps 
directly follows from the set of
physical requests made to cloud storage, the size of the delta for the
position map reveals how many real requests were
included in an epoch---padded requests do not lead to position map
updates. \sys{} thus pads the map delta to the maximum number of
entries that could have changed  in an epoch (i.e., the read batch
size times the number of read batches, plus the size of the
single write batch).

%% file: security-current.tex
\section{System Security}
\label{sec:security}

We now outline \sys's security guarantees, deferring a formal treatment to
Appendix~\ref{sec:formal-security}.
To the best of our knowledge, we are the first to formalize the notion of
crashes in the context of oblivious RAM.

\par\textbf{Model}
We express our security proof within the Universal Composability (UC)
framework~\cite{canetti01}, as it aligns well with the needs of modern
distributed systems: 
a UC-secure system remains UC-secure under concurrency or if
composed with other UC-secure systems.
Intuitively, proving security in the UC model proceeds as
follows. First, we specify an \emph{ideal functionality} \idealF
that defines the expected functionality of the protocol for both correctness and security.
For instance, \sys requires that the execution
be serializable, and that only the frequency of read and write batches
be learned. We must
ensure that the real protocol provides the same functionality to
honest parties while leaking no more information than
\idealF would.
To establish this, we consider two different worlds: one where the real
protocol interacts with an adversary \adv, and one where \idealF interacts with
\sdv[\adv], our best attempt at simulating \adv.
\adv's transcript---including its inputs, outputs, and randomness---and
\sdv[\adv]'s output are given to an environment \env, which can also observe
all communications within each world.
\env's goal is to determine which world contains the real protocol.
To prompt the worlds  to diverge, \env can delay and reorder messages, and even
control external inputs (potentially causing
failures).
Intuitively, \env represents anything external to the protocol, such as
concurrently executing systems.
We say that the real protocol is secure if,
for any adversary \adv, we can construct \sdv[\adv] such that \env
can never distinguish between the worlds.

\par\textbf{Assumptions}
The security of \sys relies on four assumptions.
\one Canonical Ring ORAM is linearizable 
\two MVTSO generates serializable executions.
\three The network will retransmit dropped packets. The adversary learns of the retransmissions, but nothing more.
\par\textbf{Ideal Functionality}
To define the ideal functionality \ideal, recall that the proxy is considered trusted while interactions with the cloud storage are not.
This allows \ideal to replace the proxy and intermediate between clients and the storage server,
performing the same functions as the proxy (we do not try to hide the concurrency/batching logic).
We must, however, define \ideal to obliviously hide data values and access patterns.
To this end, when the proxy logic finalizes a batch, \ideal simply informs the storage server that it is executing a read or write batch.
Since \ideal is a theoretical ideal, we allow it to manage all storage internally, so it then updates its local storage and furnishes the appropriate response to each client.

In this setup, modeling proxy crashes is straightforward.
Crashes can occur at any time and cause the proxy to lose all state.
So, on an external input to crash, \ideal simply clears its state.
Since we accept that \adv may learn of proxy crashes, \ideal also sends a message to the storage server that it has crashed.

\par\textbf{Proof Sketch}
The correctness of the system is straightforward, as \ideal behaves much the same as the proxy.

To prove security, we must demonstrate that, for any algorithm \adv defining the behavior of the storage server, we can accurately simulate \adv's behavior using only the information provided by \ideal.
Note that the simulator \sdv[\adv] can run \adv internally, as \adv is simply an algorithm.
Thus we can define \sdv[\adv] to operate as follows.
When \sdv[\adv] receives notification of a batch, it constructs a parallel ORAM batch from uniformly random accesses of the correct type.
It provides these accesses to \adv and produces \adv's response.

The security of this simulation hinges on two key properties:
\one the caching and deduplication logic do not affect the distribution of physical accesses, and
\two the physical access pattern of a parallelized batch is entirely determined by the physical accesses proscribed by sequential Ring ORAM for the same batch.
The first follows from Ring ORAM's guarantee that each access will be an independent uniformly random path---removing an independently-sampled element does not change the distribution of the remaining set.
The second follows from the parallelization procedure simply aggregating all accesses and performing all reads followed by all writes.

These properties ensure that the random access pattern produced by \sdv[\adv] is identical to the access pattern produced by the proxy when operating on real data.
Thus the simulated \adv must behave exactly as it would when provided with real data, and produce indistinguishable output.

%% file: implementation-current.tex
\section{Implementation}
\label{sec:implementation}

\begin{figure}
\centering
\includegraphics[width=1.0\linewidth]{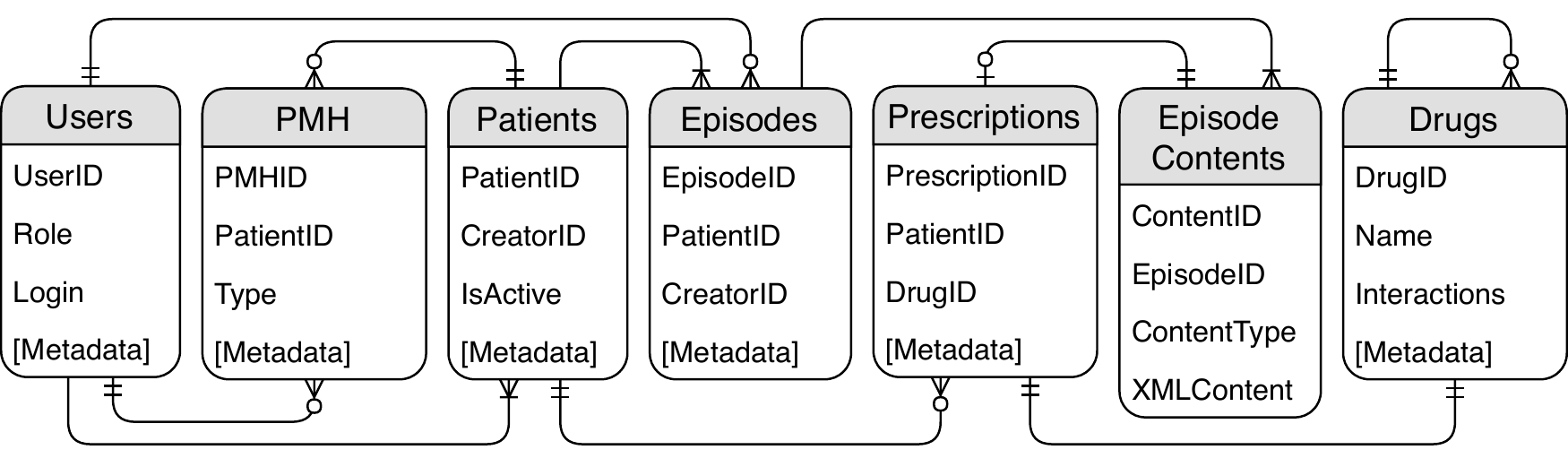}
    \caption{FreeHealth Database Architecture}
    \label{fig:freehealth}
\end{figure}

Our prototype consists of 41,000 lines of Java code. We use the Netty
library for network communication (v4.1.20), Google protobuffers for
serialization (v3.5.1), the Bouncy Castle library (v1.59) for
encryption, and the Java MapDB library (v3) for persistence. We
additionally implement a non-private baseline (NoPriv).  NoPriv shares
the same concurrency control logic (TSO), but replaces the proxy data
handler with non-private remote storage. NoPriv neither batches nor
delays operations; it buffers writes at the local proxy until commit,
and serves writes locally when possible.

%% file: evaluation-current.tex
\section{Evaluation}
\label{sec:evaluation}

\sys{} leverages the flexibility of transactional commits to mitigate the overheads of ORAM. To quantify the benefits and limitations of this
approach, we ask:
\begin{enumerate}
    \item How much does \sys{} pay for privacy? (\S\ref{subsec:end})
    \item How do epochs affect these overheads? (\S\ref{subsec:epochs})
    \item Can \sys{} recover efficiently from failures? (\S\ref{subsec:durability})
\end{enumerate}

\par \textbf{Experimental Setup} The proxy runs on a c5.xlarge
Amazon EC2 instance (16 vCPUs, 32GB RAM), and the storage 
on an m5.4xlarge instance (16 vCPUs, 64GB RAM).
The ORAM tree is configured with $Z=100$ and optimal values of $S$
and $A$ (respectively, 196 and 168)~\cite{ren15ringoram}. We report
the average of three 90 seconds runs (30 seconds ramp-up/down).

\par \textbf{Benchmarks} We evaluate the performance of our system using
three applications: TPC-C~\cite{difallah2013oltpbench,tpcc}, SmallBank~\cite{difallah2013oltpbench}, and FreeHealth~\cite{freehealthfrench,freehealth}. Our microbenchmarks
use the YCSB~\cite{cooper2010ycsb} workload generator.
\textbf{TPC-C}, the defacto standard for OLTP workloads,
simulates the business logic of e-commerce
suppliers. We configure TPC-C to
run with 10 warehouses~\cite{xie2015callas}.
In line with prior transactional key-value stores~\cite{su2017tebaldi},
we use a separate table as a secondary index on the \texttt{order}
table to locate a customer’s latest order in the \texttt{order status}
transaction, and on the \texttt{customer} table to look up customers
by their last names (\texttt{order status} and
\texttt{payment}). \textbf{Smallbank}~\cite{difallah2013oltpbench}
models a simple banking application supporting money transfers,
withdrawals, and deposits. We configure it to run with one million
accounts.  Finally, we port
\textbf{FreeHealth}~\cite{freehealthfrench,freehealth}, an
actively-used cloud EHR system (Figure~\ref{fig:freehealth}).
FreeHealth supports the business logic of medical practices and
hospitals. It consists of 21 transaction types that doctors use to
create patients and look up medical history, prescriptions, and drug
interactions.

\subsection{End-to-end Performance}
\label{subsec:end}

\begin{figure}[tb]%
\centering
    \subfloat[Throughput]
    {\label{graph:stride}
    \includegraphics[width=0.45\linewidth,valign=b]
    {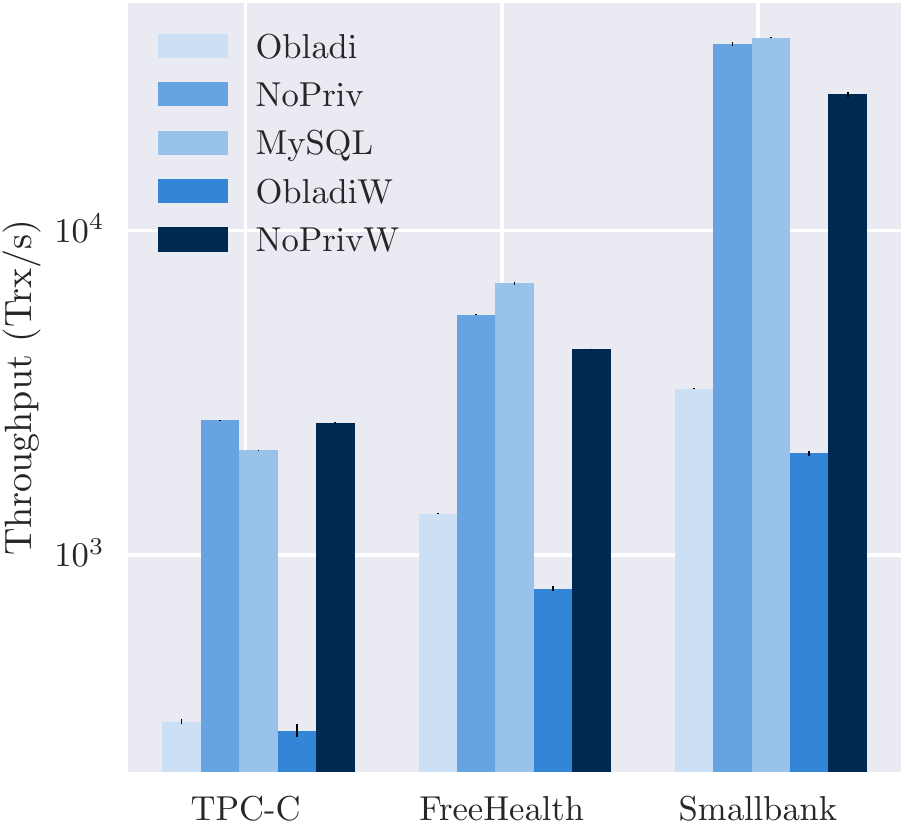}}
    \subfloat[Latency]
    {\label{graph:applicationslatency}
    \includegraphics[width=0.45\linewidth,valign=b]
    {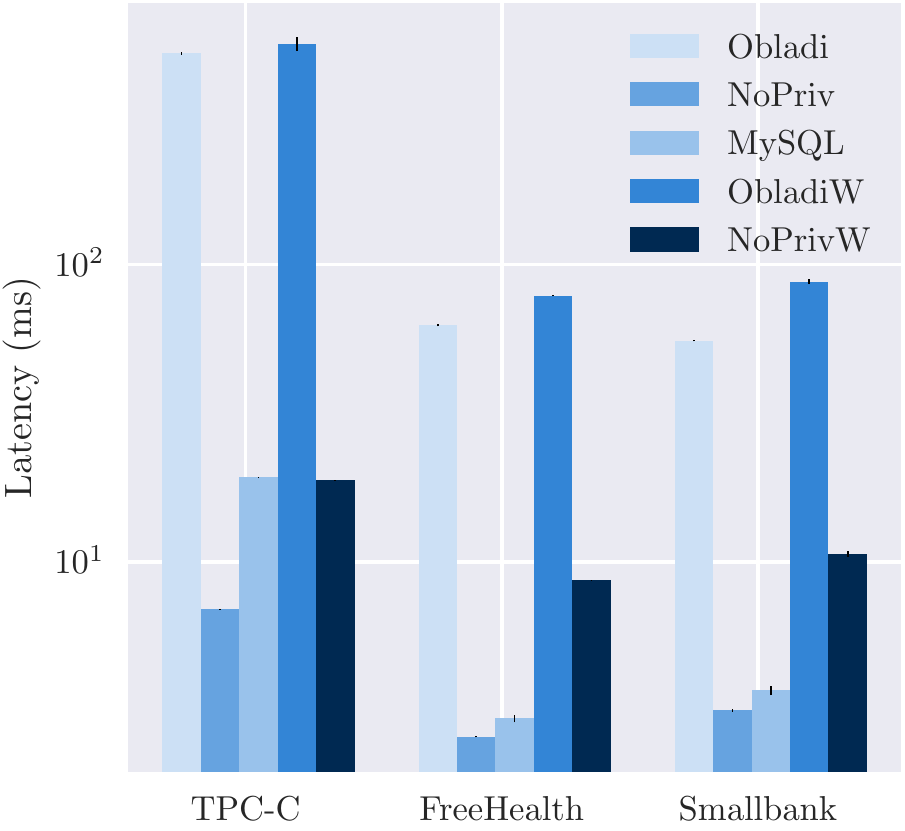}}
     \caption{Application Performance}%
     \label{fig:e2e}
\end{figure}

Figure~\ref{fig:e2e} summarizes the results from running the three
end-to-end applications in two setups: a local setup in 
which the latency between proxy and server is low (0.3ms) (\textbf{Obladi, NoPriv}), and a more realistic
WAN setup with 10ms latency (\textbf{ObladiW},  \textbf{NoPrivW}). We additionally
compare those results with a local MySQL setup. MySQL, unlike NoPriv, cannot buffer writes. We consequently
do not evaluate MySQL in the WAN setting.
\par \textbf{TPC-C} \sys{} comes within 8\texttimes\ of NoPriv's
throughput, as NoPriv is contention-bottlenecked on the high rate of
conflicts between the \texttt{new-order} and \texttt{payment}
transactions on the \texttt{district} table. NoPriv's
  performance is itself slightly higher than MySQL as the use of MVTSO
  allows for the \texttt{new-order} and \texttt{payment} transactions
  to be pipelined. In contrast, MySQL acquires exclusive locks for the
  duration of the transactions.  Latency, however, spikes to
70\texttimes\ over NoPriv because of the inflexible execution pattern
\sys{} needs for security.  Transactions in TPC-C vary heavily in
size.
Epochs must be large
enough to accommodate all transactions, and hence artificially
increase the latency of short instances. Moreover, write
operations must be applied atomically during epoch changes. For a
write batch size of 2,000, this process takes on average 340ms, further
increasing latency for individual transactions. The write-back process
also limits throughput, even preventing non-conflicting operations
from making progress (in contrast, NoPriv can benefit from writes
never blocking reads in MVTSO). Epoch changes also introduce
additional aborts for transactions that straddle epochs.
The additional 10ms latency of the WAN
setting has comparatively little effect, as the large write batch
size of TPC-C is the primary bottleneck: throughput remains within
9x of NoPrivW. Also NoPrivW's performance  does not degrade:
since MVTSO exposes uncommitted writes immediately, increasing commit latency 
does not increase contention.

\par \textbf{Smallbank} Transactions in Smallbank are more
homogeneous (between three and six operations); thus,  the length of
an epoch can be set to more closely approximate most
transactions, reducing latency overheads (17\texttimes\ NoPriv).
NoPriv is CPU bottlenecked for Smallbank; the relative throughput
drop for \sys{} is higher (12\texttimes) because of the overhead of
changing epochs and the blocking that it introduces. Transaction
dependency tracking becomes a bottleneck in NoPriv, 
resulting in a 15\% throughput loss over MySQL. Increasing
latency between proxy and storage causes both systems' throughput to
drop. ObladiW's 35\% drop is due to the increased duration of
epoch changes (during which no other transactions
can execute) while NoPrivW's 30\% drop stems from the larger
dependency chains that arise from the relatively long commit phase. 

\par \textbf{FreeHealth} Like SmallBank, FreeHealth consists of fairly
short transactions and can thus choose a fairly small epoch (five read
batches), reducing the impact on latency  (20\texttimes\ NoPriv). Unlike
Smallbank, however, FreeHealth consists primarily of read operations,
and so it can choose a much smaller write batch (200), minimizing the
cost of epoch changes and maximizing throughput (only a 4\texttimes\
drop over NoPriv and a 5.5\texttimes\ over NoPrivW for ObladiW).  Both NoPriv and \sys are contention-bottlenecked on the
creation of {\em episodes}, the core units of EHR
systems that encapsulate prescriptions, medical history, and patient
interaction.

\subsection{Impact of Epochs}
\label{subsec:epochs}

Though epochs create blocking and cause aborts, they are
key to reducing the cost of accessing ORAM, as they allow to  \one~securely parallelize the  ORAM and
\two~delay and buffer bucket writes.  To quantify epochs' impact on performance as a function of their
size and the underlying storage properties, 
\begin{figure}[tb]%
\centering
    \subfloat[Parallelism (Batch Size 500)]
    {\label{graph:parallel}
    \includegraphics[width=0.45\linewidth,valign=b]
    {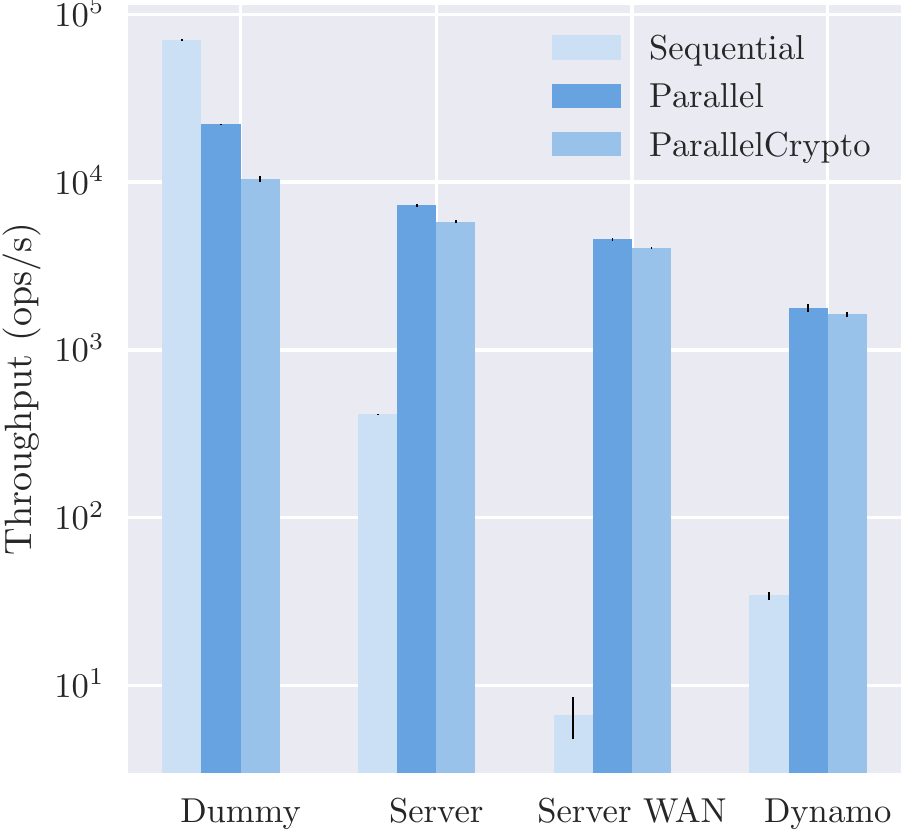}}
    ~~
    \subfloat[Batch Size Throughput]
    {\label{graph:stride}
    \includegraphics[width=0.45\linewidth,valign=b]
    {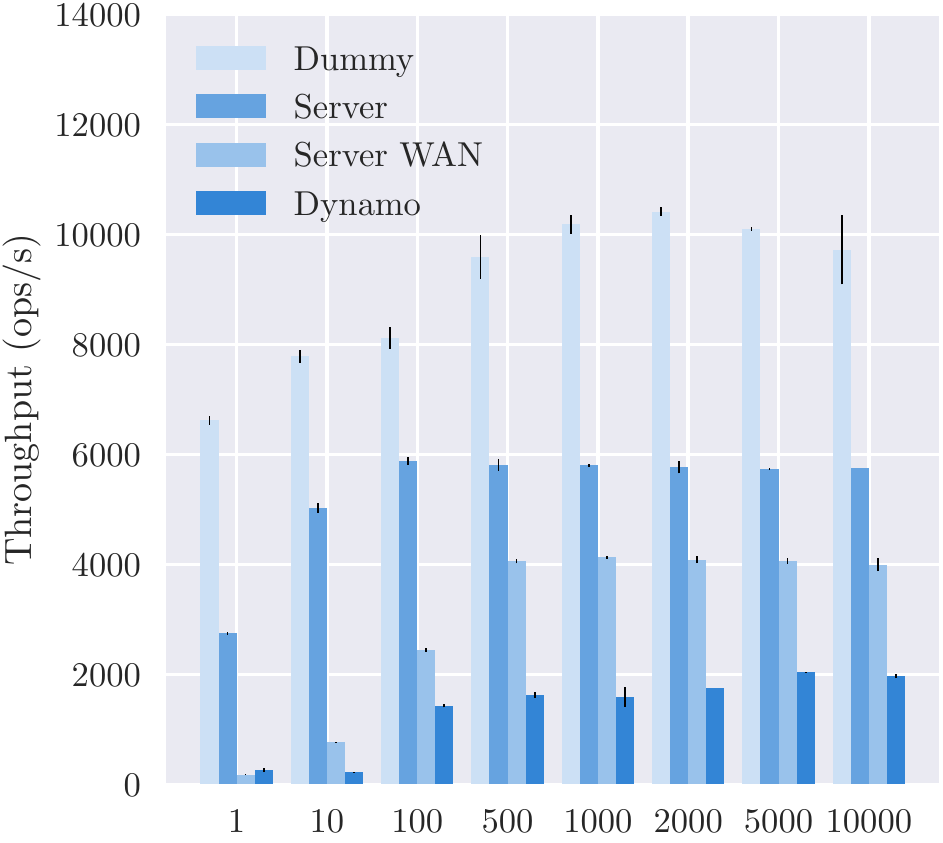}}
    \\[-0.5ex]
    \subfloat[Batch Size Latency] 
    {\label{graph:latency}
    \includegraphics[width=0.45\linewidth,valign=b]
    {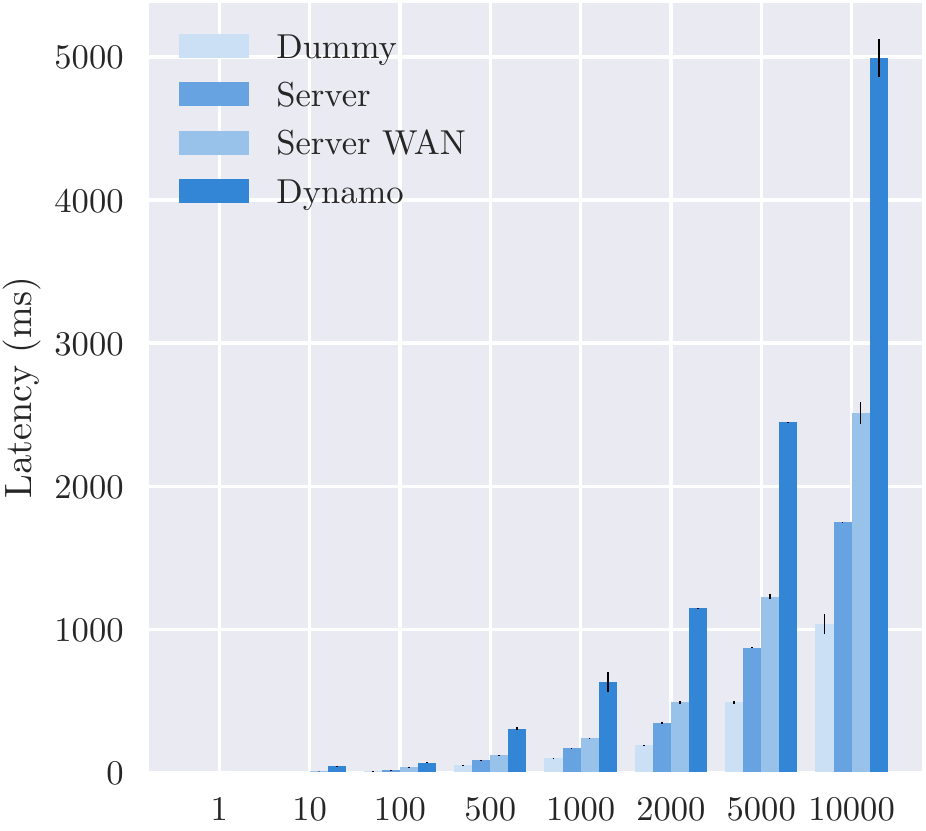}}
    ~~
    \subfloat[Delayed Visibility]
    {\label{graph:delayed}
    \includegraphics[width=0.45\linewidth,valign=b]
    {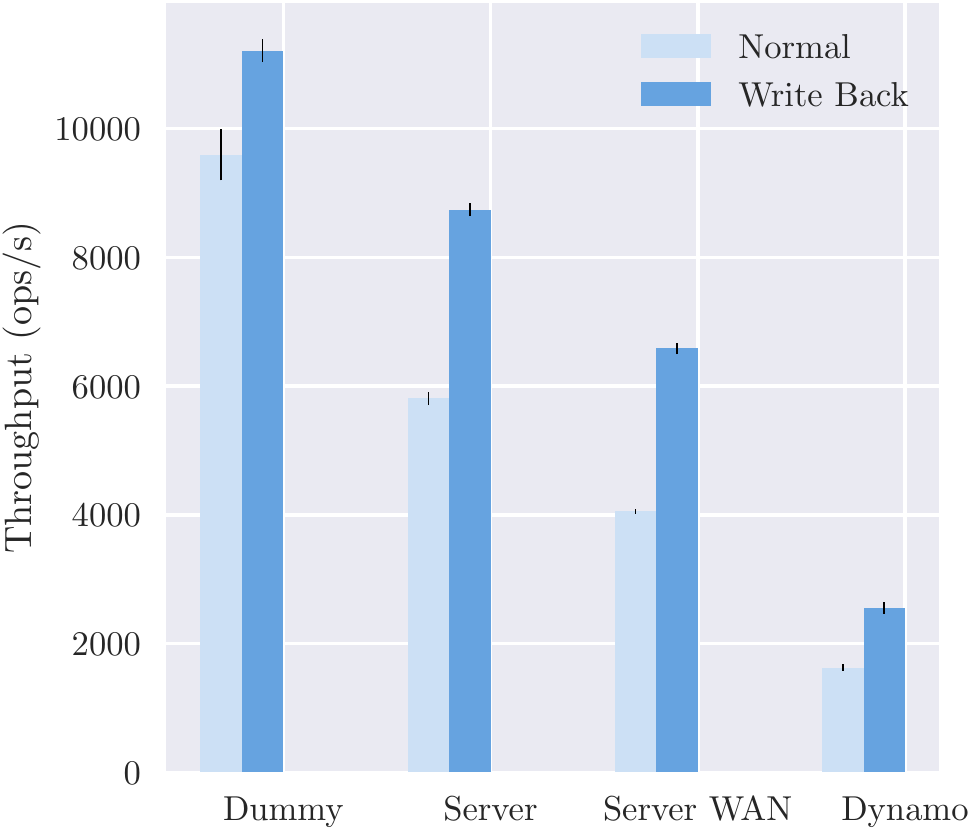}}
    \\[-0.5ex]
    \subfloat[Epoch Size Impact - ORAM]
    {\label{graph:batch}
    \includegraphics[width=0.45\linewidth,valign=b]
    {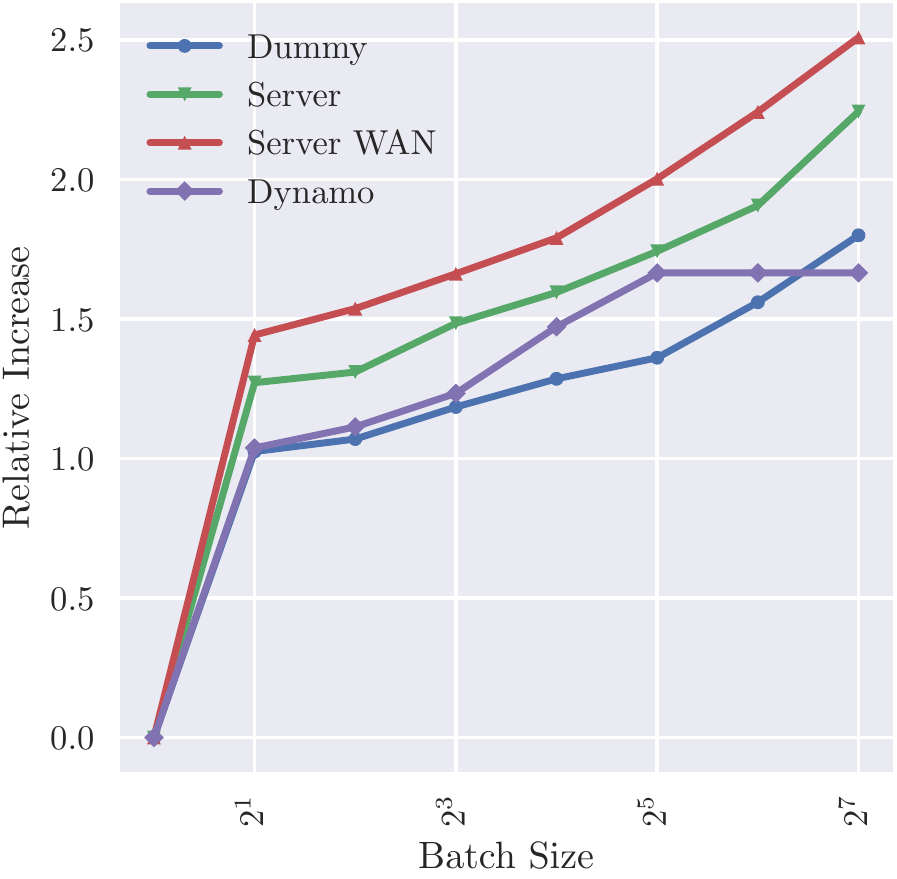}}
    ~~
    \subfloat[Epoch Size Impact - Proxy]
    {\label{graph:proxy}
    \includegraphics[width=0.45\linewidth,valign=b]
    {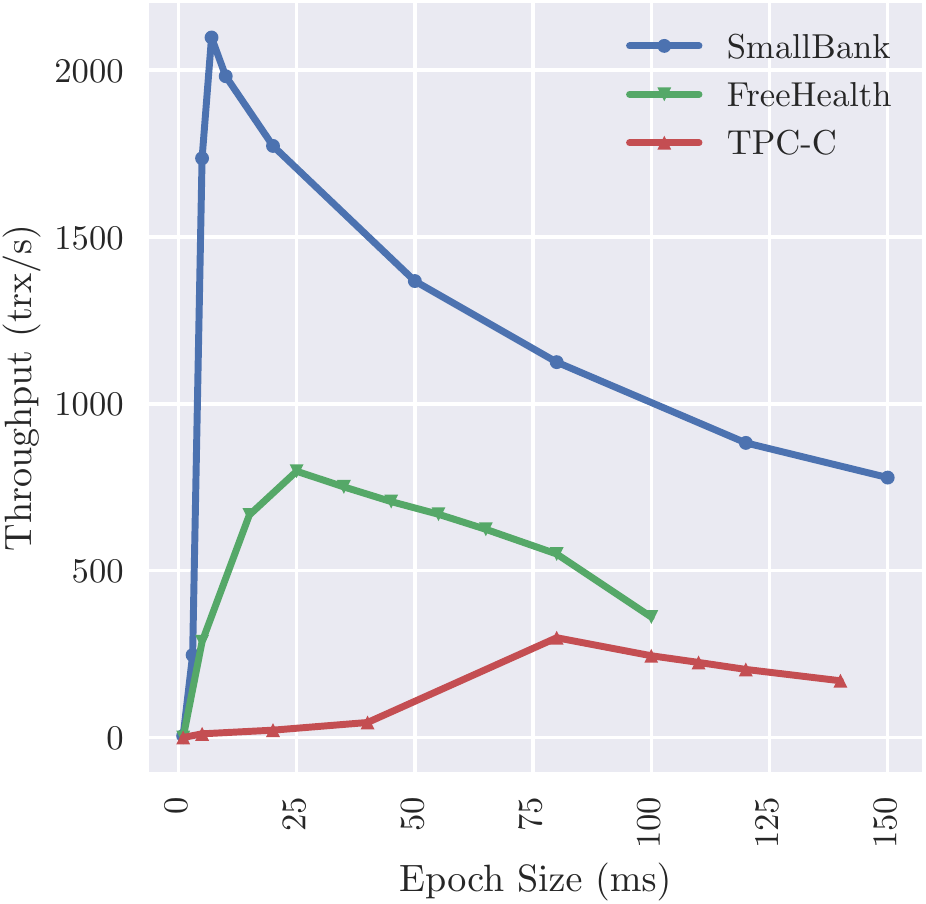}}
    \caption{Performance impact of various features}
\end{figure}
we instantiate an ORAM with 100K objects and choose three different storage backends:
a local dummy (storing no real data) that responds to all reads with a static value and ignores writes (\texttt{dummy});
a remote server backend with an in-memory hashmap (\texttt{server},
ping time 0.3ms) and a remote WAN server backend with an in-memory hashmap (\texttt{server WAN}, ping time 10ms);
and DynamoDB (\texttt{dynamo}, provisioned for 80K req/s, read ping 1ms, write 3ms).

\par \textbf{Parallelization} We first focus on the performance impact
of parallelizing Ring ORAM (ignoring other
optimizations). Graph~\ref{graph:parallel} shows that, unsurprisingly,
the benefits of parallelism increase with the latency of individual
requests.  Parallelizing the ORAM for \texttt{dummy}, for instance,
yields no performance gain; in fact, it results in a 3\texttimes\
slowdown (from 72K req/s to 24K req/s). Sequential Ring ORAM on 
\texttt{dummy} is CPU-bound on metadata computation (remapping paths,
shuffling buckets, etc.), so adding coordination mechanisms to
guarantee multi-level serializability only increases the cost of
accessing a bucket.  As storage access latency increases and the ORAM
becomes I/O-bound, the benefits of parallelism become more
salient. For a batch size of 500, throughput increases by
12\texttimes\ for \texttt{server}, as much as 51\texttimes\ for
\texttt{dynamo}, and 510\texttimes for \texttt{WAN server}. The available
parallelism is a function of both the size/fan-out of the tree and the
underlying resource bottlenecks of the proxy.
Graph~\ref{graph:stride} captures the parallelization speedup for both
intra- and inter-request parallelism, while Graph~\ref{graph:stride}
quantifies the latency impact of batching. The parallelization speedup
achieved for a batch size of one captures intra-request parallelism:
the eleven levels of the ORAM can be accessed concurrently, yielding
an 11\texttimes\ speedup.  As batch sizes increase, \sys{} can
leverage inter-request parallelism to process non-conflicting physical
operations in parallel, with little to no impact on
latency. \texttt{Dynamo} peaks early (at 1750 req/s) because its
client API uses blocking HTTP calls, and \texttt{dummy}'s storage
eventually bottlenecks on encryption, but \texttt{server} and \texttt{WAN
server} are more interesting. Their throughput is limited by the
physical and data dependencies on the upper levels of the tree (recall
that paths always conflict at the root (\S\ref{sec:parallel})).

\par \textbf{Work Reduction} 
To amortize ORAM overheads across a large number of operations, \sys{}
relies on delayed visibility to buffer bucket writes until the end of
an epoch, when they can be executed in parallel, discarding
intermediate writes. Reads to those buckets are directly served from
the proxy, reducing network communication and CPU work (as encryption
is not needed). Graph~\ref{graph:delayed} shows that 
enabling this optimization for an epoch of eight batches (a setup
suitable for FreeHealth and TPC-C) yields a 1.5\texttimes\ speedup on
both \texttt{dynamo} and the server, a 1.6\texttimes\ speedup on the
WAN server, but only minimal gains for \texttt{dummy}
(1.1\texttimes). When using a small number of
batches, throughput gains come primarily from
combining duplicate operations in buckets near the top of the
tree. For example, the root bucket is written 27 times in an epoch of size
eight (once per eviction, every 168 requests).  As these
operations conflict, they must be executed sequentially and quickly
become the bottleneck (other buckets have fewer operations to
execute). Our optimization lets \sys{}  write the
root bucket only once, significantly reducing latency and thus
increasing throughput. As epochs grow in size,
increasingly many buckets are buffered locally until the
  end of the epoch (\S\ref{sec:parallel}), allowing reads to
be served locally and further reducing I/O with the storage.  Consider
Graph~\ref{graph:batch}: throughput increases almost logarithmically;
metadata computation eventually becomes a bottleneck for
\texttt{dummy}, while \texttt{server} and \texttt{server WAN}
eventually run out of memory from storing most of the tree (our AWS
account did not allow us to provision \texttt{dynamo} adequately for
larger batches). Larger epochs reduce the raw amount of work per  operation: with one batch, \sys requires 41 physical
requests per logical operation, but only requires 24 operations with
eight batches.
  For real transactional workloads, however,  epochs are
   not a silver bullet. Graph~\ref{graph:proxy} suggests that 
  applications are very sensitive to identifying the right epoch duration:
  too short and transactions cannot make progress, repeatedly aborting;
  too long and the system will remain unnecessarily idle.

\subsection{Durability}
\label{subsec:durability}

\begin{figure}[tb]
    \vspace*{-1ex}
    \subfloat[Checkpoint Frequency (100K)]
    {\label{graph:checkpointing}
    \parbox[c][3.5cm]{0.47\linewidth}{
      \centering
      \includegraphics[width=\linewidth]{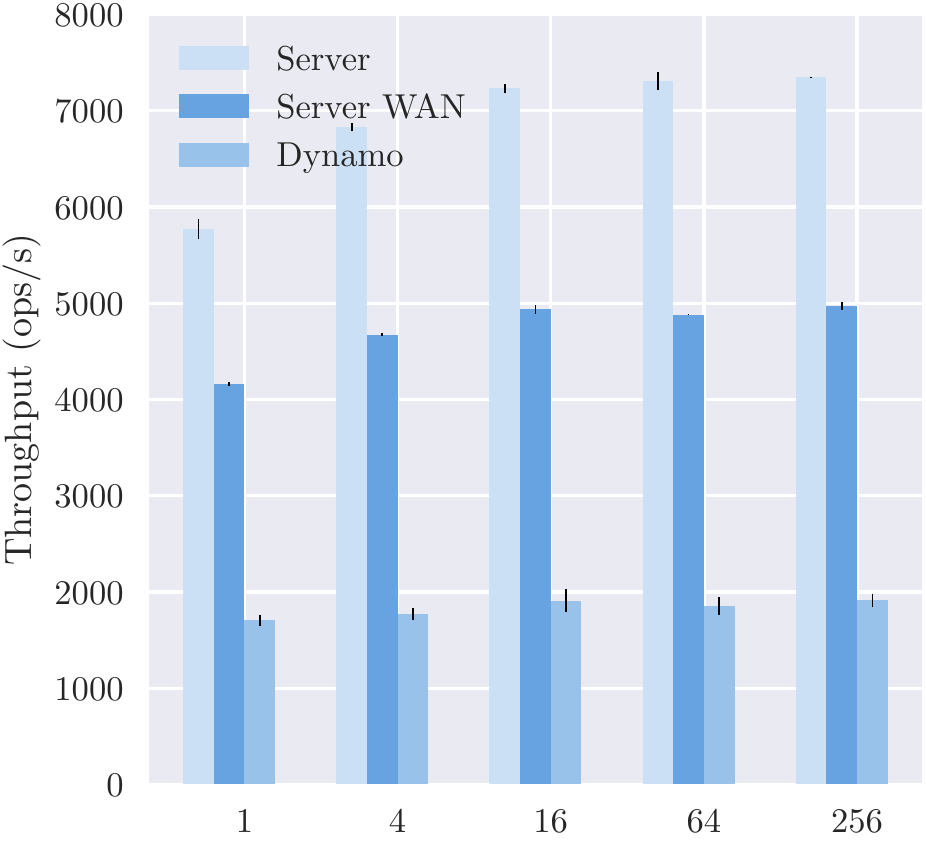}
    }}
    ~~
    \subfloat[Server Wan Recovery Time (ms)]
    {\label{graph:recovery}
    \parbox[c][3.5cm]{0.47\linewidth}{
      \centering
      \input{recovery}

    }
    }
    \caption{Durability}
\end{figure}

Table~\ref{graph:recovery} quantifies the efficiency of failure
recovery and the cost it imposes on normal execution for ORAMS of
different sizes (we show space results for only the WAN server as
Dynamo follows a similar trend).  During normal execution, durability
imposes a moderate throughput drop (from 0.83\texttimes\ for 10K to
0.89\texttimes\ for 1M).  This slowdown is due to the need to
checkpoint client metadata and to synchronously log read paths to
durable storage before reading. As seen in
Graph~\ref{graph:checkpointing}, computing diffs mitigates the impact
of checkpointing.  Recovery time similarly increases as the ORAM
grows, from 1.5s to 6.1s (Table~\ref{graph:recovery}, \emph{RecTime}).
The costs of decrypting the position and permutation maps (\emph{Pos}
and \emph{Perm}) are low for small datasets, but grow linearly
with the number of keys.  Read path logging (\emph{Paths}) instead
starts much larger, but grows only with the depth of the tree.

%% file: recovery.tex
\footnotesize{
 \setlength\tabcolsep{2pt}
\begin{tabular}[b]{|c|c|c|c|}
\hline
             &  10K & 100K &   1M \\ \hline
    Levels   &    7 &   11 &   14 \\ \hline
    Slowdown & 0.83 & 0.88 & 0.89 \\ \hline 
    RecTime  & 1452 & 2604 & 6080 \\ \hline
    Network  &  182 &  681 &  848 \\ \hline
    Pos      &    8 &   74 & 1610 \\ \hline
    Perm     &   15 &  218 & 1424 \\ \hline
    Paths    &  864 & 1104 & 1341 \\ \hline
\end{tabular}
}

%% file: related-current.tex
\section{Related Work}
\label{sec:related}

\par \textbf{Batching} \sys{} amortizes ORAM costs by grouping
operations into epochs and committing at epoch
boundaries. Batching can mitigate expensive security
primitives, e.g., it reduces server-side computation in
private information retrieval (PIR)
schemes~\cite{gupta16popcorn,beimel04reducing,lueks2014sublinear,ishai2006batchcodes},
amortizes the cost of shuffling networks in
Atom~\cite{kwon2017atom} and the cost of verifying integrity in Concerto~\cite{arasu2017concerto}. Changing when operations output commit is
a popular performance-boosting technique: it yields
significant gains for state-machine
replication~\cite{ports2015specpaxos,
  kotla2010zyzzyva,kapritsos2012eve}, file
systems~\cite{nightingale2008rethink}, and transactional
databases~\cite{mehdi17occult,crooks2017seeing,tu2013silo}.

\par \textbf{ORAM parallelism} \sys{} extends recent work on parallel
ORAM constructions~\cite{lorch2013shroud,williams12privatefs,bindschaedler2015curious} to extract parallelism
both \textit{within} and \textit{across} requests.
Shroud~\cite{lorch2013shroud} targets intra-request parallelism by
concurrently accessing different levels of tree-based ORAMs. Chung et
al~\cite{boyle16pram} and PrivateFS~\cite{williams12privatefs} instead
target inter-request parallelism, respectively in
tree-based~\cite{shi2011oblivious} and
hierarchical~\cite{williams2008castles} ORAMs. Both works execute
requests to distinct logical keys concurrently between reshuffles or
evictions and deduplicate concurrent requests for the same key to
increase parallelism. \sys{} leverages delayed visibility to separate
batches into read and write phases, extracting concurrency both within
requests and across evictions. Furthermore, \sys{} parallelizes across
requests by deduplicating requests at the trusted proxy.

ObliviStore~\cite{stefanov2012towards} and
Taostore~\cite{sahin2016taostore} instead approach parallelization by
focusing on asynchrony. ObliviStore~\cite{stefanov2012towards}
formalizes the security challenges of scheduling requests
asynchronously; the  oblivious scheduling mechanism  that it presents
for that model however is computationally expensive and requires a
large stash, making ObliviStore unsuitable for implementing ACID
transactions. Like ObliviStore, Taostore leverages asynchrony to parallelize Path
ORAM~\cite{stefanov13pathoram}, a tree-based construction from which
Ring ORAM descends. Taostore, however, targets a different threat
model: it assumes both that requests must be processed immediately,
and that the timing of responses is visible to the adversary. Request
latencies thus necessarily increase linearly with the number of
clients~\cite{williams12privatefs}.

\par \textbf{Hiding access patterns for non-transactional systems}
Many systems seek to provide access pattern protections for analytical
queries: Opaque~\cite{zheng17opaque} and
Cipherbase~\cite{arasu2013cipherbase} support oblivious operators for
queries that scan or shuffle full tables. Both rely on hardware
enclaves for efficiency: Opaque runs a query optimizer in
SGX~\cite{SGX}, while Cipherbase leverages secure co-processors to
evaluate predicates more efficiently. Others seek to
hide the parameters of the query rather than the query itself:
Olumofin et al.~\cite{olumofin2010pirsql} do it via multiple rounds
of keyword-based PIR operations~\cite{chor97keyword};
Splinter~\cite{wang17splinter} reduces the number of round-trips
necessary by mapping these database queries to function secret sharing
primitives.  Finally, ObliDB~\cite{saba2017oblidb} adds support for
point queries and efficient updates by designing an oblivious B-tree
for indexing. The concurrency control and
recovery mechanisms of all these approaches introduce timing channels
and structure writes in ways that leak access patterns~\cite{arasu2013cipherbase}.

\par \textbf{Encryption} Many commercial
systems offer the possibility to store encrypted data~\cite{alwaysencrypted,dynamoencrypted}.
Efficiently executing data-dependent queries like joins, filters,
or aggregations without knowledge of the plaintext is
challenging: systems like CryptDB~\cite{popa11cryptdb},
Monomi~\cite{tu2013monomi}, and Seabed~\cite{papadimitriou2016seabed}
tailor encryption schemes to allow executing certain queries
directly on encrypted data. Others leverage trusted hardware
~\cite{bajaj2011trusteddb}. In contrast, executing transactions on encrypted data
is straightforward: neither concurrency control
nor recovery requires knowledge of the plaintext data.

%% file: conclusion-current.tex
\section{Conclusion}
\label{sec:conclusion}
This paper presents~\sys{}, a system that, for the first time, considers the
security challenges of providing ACID transactions without revealing access patterns.
\sys{} guarantees security and durability at moderate cost through a simple observation:
transactional guarantees are only required to hold for committed transactions. 
By delaying commits until the end of epochs, \sys{} inches closer to
providing practical oblivious ACID transactions.

\par \textbf{Acknowledgements} We thank our shepherd, Jay Lorch, for
his commitment to excellence, and the anonymous reviewers for their
helpful comments. We are grateful to Sebastian Angel, Soumya Basu,
Vijay Chidambaram, Trinabh Gupta, Paul Grubbs, Malte Schwarzkopf,
Yunhao Zhang, and the MIT PDOS reading group for their feedback.  This
work was supported by NSF grants CSR-1409555 and CNS-1704742, and an
AWS EC2 Education Research grant.

%% file: integrity-current.tex
\section{Ensuring Data Integrity in \sys}
\label{sec:integrity}

As we described in \S\ref{sec:threat}, we assume the untrusted storage
server is \emph{honest-but-curious}.  In many cases this is a very
strong assumption that system operators may not be happy to make.  We
can remove this requirement with the use of Message Authentication
Codes (MACs) and a trusted counter---used to ensure freshness---that
persists across crashes.  We describe this technique here.

When we assumed the server was honest-but-curious, we assumed it could
deny service, but would otherwise correctly respond to all queries.
In order to remove this assumption while maintaining security, we must
create a means to identify if the storage server returns incorrect
data, thus reducing them to DoS attacks.  To do this, the proxy must
verify that the returned value is the value \one most recently written
\two by the proxy \three to the specified location.

We can guarantee \two using MACs.  At initialization, the proxy
generates a secret MAC key (in addition to its secret encryption key)
and attaches a MAC to every piece of data it stores on the cloud
server.  This allows the proxy to verify that the cloud server did not
modify the data or manufacture its own.

By themselves, MACs do not guarantee \one and \three, as the cloud
server can provide an old copy of the data or valid data from a
different location, both of which will have valid MACs.  We
additionally need to include a unique identifier that the proxy can
easily recompute.  For data that is written at most once per epoch,
this unique identifier can be the pair of epoch, ORAM location.  Due
to Ring ORAM's deterministic eviction algorithm, the proxy can compute
the epoch during which any given block was most recently written
knowing only the current epoch counter and the early reshuffle table.

There is exactly one value which is written multiple times per epoch:
each read batch, of which there may be many per epoch, logs the
accessed locations.  This means the counter associated with those
writes must uniquely identify the read batch, not just the epoch.  In
fact, since every epoch has the same number of read batches, a read
batch counter is sufficient for all values.

\textbf{Handling Crashes} The above modifications are sufficient to
guarantee integrity if the proxy never crashes.  When the proxy
crashes, however, it needs information from the cloud storage to
recover.  To guarantee integrity---in particular freshness---of the
recovery data, the epoch/read batch counter we describe above must
persist in a trustworthy fashion across failures.  Perhaps the easiest
way to implement this requirement is to store the counter on a small
amount of nonvolatile storage locally on the proxy, but any
trustworthy and persistent storage mechanism is sufficient.

This, of course, raises the question of when to update this
trustworthy persistent counter.  Once the update occurs, a recovering
proxy will expect the cloud storage to provide data associated with
that counter value.  This means that the counter must be updated after
writing to cloud storage.  Because a recovering proxy will be unaware
of the newly-written data until the counter is updated, we do not
consider the write complete until the counter is properly updated.  As
usual, if the proxy crashes while a write is in-progress, the write is
simply rolled back.

As long as the storage server cannot learn anything from incomplete
writes, our new strategy is entirely secure.  Because the timing of
\sys's writes is completely deterministic and their locations are
determined entirely by the locations of prior reads, the fact that a
write has aborted does not inherently leak any information.  The
contents of the write can, however, leak information if we are not
careful.  Most data in the system is already encrypted, but one value
is not: the read logs written during read batches.  Previously we had
no need to encrypt these as the write operation was atomic and the
cloud server was to immediately learn all data contained in the write.
However, the write is no longer atomic; the proxy can crash after
sending data to the cloud server but before updating its trusted
counter.  In this case the storage server may withhold that data on
recovery without detection and learn whether the proxy accessed the
same locations after recovery is complete.  To fix this leak, we
encrypt the read batch logs in the cloud and update the counter after
writing the log but before reading any values.  That way the cloud
storage gains no information about what data will be read until after
the write is complete, at which point the proxy will always replay the
read if a crash occurs.  Thus we have removed the leakage.

As we will see in Appendix~\ref{sec:formal-security}, these
modifications are sufficient to guarantee both confidentiality and
integrity (though obviously not availability) even against an
arbitrarily malicious cloud storage server.

%% file: formal-security-current.tex
\section{Formal Security}
\label{sec:formal-security}

We now provide formal security definitions and proofs for \sys.
As we discuss in \S\ref{sec:security}, we use the Universal Composability (UC) framework~\cite{canetti01}.
The UC framework requires us to specify an ideal functionality \ideal that defines what it means for \sys to be secure.
We must then prove that, for every possible adversarial algorithm \adv specifying the behavior of the storage server, we can simulate \adv's behavior when interacting only with \ideal.

We prove security of the scheme including the modification in Appendix~\ref{sec:integrity} and do not assume the cloud storage provider is trusted for integrity.
As the MACs and counters are only used to verify integrity and freshness of data, they are unnecessary if the cloud server is being honest.
As we will see below, removing them---as we do in our implementation---does not impact security in this case.

We also noted in Appendix~\ref{sec:integrity} that the proxy requires a trusted epoch counter that persists across crashes.
This could be implemented as an integer in local non-volatile storage that the proxy updates with each epoch,
it could be implemented by trusting the cloud storage for integrity and saving it there, or other means.
We abstract away this detail by providing the \sys protocol with access to \idealEpoch, an ideal functionality that provides access to this counter.

\subsection{Ideal Functionality}

We begin by noting that \sys's proxy acts as a trusted central coordinator that performs publicly-known logic on private data.
As this is essentially the role played by any ideal functionality, we simply subsume the proxy into \ideal.
Moreover, some of the proxy's behavior, like the fact that it deduplicates and caches accesses, pads under-full batches,
is public information, meaning \ideal can explicitly perform exactly the same operations.

In \S\ref{sec:architecture} we describe the proxy as consisting of a concurrency control unit and a data manager, which itself contains a batch manager and ORAM executor.
As the concurrency control and batch management functionalities do not inherently leak any information, we define \ideal in terms of those operations.
In particular, we let \idealProxy represent this functionality.
\idealProxy is defined as providing the exact functionality of the concurrency control unit and batch manager as described in \S\ref{sec:architecture} and \S\ref{sec:proxy}.
\idealProxy has the following ways to interface with \ideal:
\begin{itemize}
  \item \ideal can supply \idealProxy with an input from a client (start, read, write, or commit).
  \item \idealProxy can produce a read batch of logical data blocks.
    The batch need not be full, meaning it may contain fewer than the maximum number of reads for a batch.
    \ideal can then respond with the requested blocks.
  \item \idealProxy can produce a write batch of logical data blocks. The batch need not be full.
    \ideal can then respond confirming the writes have completed.
  \item \idealProxy can specify an epoch has ended and transactions should commit.
    \ideal can then respond with confirmation.
  \item \ideal can clear \idealProxy's internal state, representing a crash.
\end{itemize}
\idealProxy can additionally send a messages directly to clients.

\textbf{Modeling Crashes}
In the real system the proxy can crash at any time.
As all state except the cryptographic keys (and possibly trusted counters) is considered volatile,
it does not matter when during a local operation the proxy crashes, as every piece of that operation is lost regardless.
We can therefore simplify the ideal functionality by allowing for crashes both between requests and immediately prior to
any operation within a request that either leaves the proxy (e.g., writing to cloud storage) or persists across crashes (e.g., updating the trusted epoch counter).

To model any possible crash, we control the timing of the crashes through a Crash Client.
\ideal queries the Crash Client immediately prior to any relevant action and waits for a reply.
The Crash Client then waits for a prompt from the environment, which it forwards to \ideal, telling it to proceed or crash.
Additionally, the Crash Client---again at the prompting of the environment---can issue a ``crash'' command independently between requests.

We provide the full specification for \ideal in Algorithm~\ref{alg:ideal-func}, which references \idealProxy.
For notational clarity, we do not explicitly specify every call to the Crash Client.
Instead any operation prefixed by {\tinydag} notifies the Crash Client before executing and crashes if instructed.
Note that it is possible to crash while recovering from a crash.

\begin{algorithm}[tb!]
  \SetKwBlock{Init}{Initialize}{end}
  \SetKwBlock{Fn}{function}{end}
  \SetKwFunction{Recover}{crashRecover}

  \KwData{$D = \text{DatabaseState}$}
  \KwData{Counters $c_e = 0~;~c_b = 0$}
  \BlankLine
  \Init{
    Initialize \idealProxy \;
    Begin epoch \;
  }
  \BlankLine
  \Receive($m$ from client \client) {
    Forward $(m, \client)$ to \idealProxy \;
  }
  \BlankLine
  \Receive(``\textsf{read-batch}$[\mathit{blks}]$'' from \idealProxy) {
    {\tinydag}Send ``\textsf{read-batch-init}'' to \adv, wait for \OK \;
    {\tinydag}$c_b \leftarrow c_b + 1$ \;
    {\tinydag}Send ``\textsf{read-batch-read}'' to \adv, wait for \OK \;
    Read $\mathit{blks}$ from $D$ \;
    Respond to \idealProxy with results \;
  }
  \BlankLine
  \Receive(``\textsf{write-epoch}$[\mathit{data}]$'' from \idealProxy) {
    {\tinydag}Send ``\textsf{write-epoch}'' to \adv, wait for \OK \;
    {\tinydag}$c_e \leftarrow c_e + 1 ~;~ c_b \leftarrow 0$ \;
    Write $\mathit{data}$ to $D$ \;
    Confirm write/epoch completed to \idealProxy \;
  }
  \BlankLine
  \Receive(``\textsf{crash}'' from Crash Client) {
    Execute \Recover
  }
  \BlankLine
  \Fn(\Recover) {
    Send $(\text{``\textsf{crash}''}, c_e, c_b)$ to \adv \;
    Clear internal state of \idealProxy \;
    Rollback writes to $D$ since beginning of epoch $c_e$ \;
    {\tinydag}$c_e \leftarrow c_e + 1 ~;~ c_b \leftarrow 0$ \;
  }
  \BlankLine
  \parbox{\columnwidth}{\footnotesize {\tinydag}Before executing operation, notify Crash Client. On response of ``crash,'' abort operation and invoke \Recover, otherwise proceed.}

  \caption{Ideal functionality \ideal using \idealProxy.}
  \label{alg:ideal-func}
\end{algorithm}

\subsection{Security Lemmas}

In order to prove the security of \sys, we rely on two lemmas which we alluded to in \S\ref{sec:security}.

\begin{lemma}[Caching and Deduplication]
  \label{lem:cache-dedup}
  Let $D$ be any set of logical reads or writes selected independently from the current ORAM position map.
  Let $D^*$ be the set of accesses resulting from applying the proxy batch manager's caching and deduplication logic to $D$.
  The set of physical accesses needed to realize $D^*$ is identically distributed to the set of physical accesses needed to realize a uniformly random set of logical accesses of the same size.
\end{lemma}

\begin{proof}
  Since $D$ is selected independently from the current position map in the ORAM,
  Ring ORAM guarantees that the set of physical accesses needed to realize $D$ is identically distributed to that for a uniformly random set of logical reads or writes.
  $D^*$ is simply $D$ with some elements removed, so we claim that the elements removed form an unbiased sample.
  Since removing an unbiased sample from a distribution does not change the distribution, this is sufficient.

  We first note that Ring ORAM guarantees that any independently-selected logical access $d$ results in physical accesses sampled independently from the following distribution.
  First sample a uniformly random path in the tree. Then, for each bucket in that path, sample a uniformly random block from among those not read since the bucket was last written.
  Finally, read all selected blocks.

  In Ring ORAM, whenever a block is read or written, it is immediately remapped to an independent uniformly random path in the tree that determines what will be read next time it is accessed.
  The proxy batch manger's caching and deduplication logic removes access requests for any block previously accessed in this epoch.
  Each of those blocks was mapped to a new independently uniform random path when accessed.
  Moreover, when an epoch ends, the cache is completely flushed, meaning there is no (potentially-biased) caching or deduplication.

  Thus the sample of physical accesses removed by pairing $D$ down to $D^*$ must be unbiased, so $D^*$ must result in a uniformly random set of physical access paths.
\end{proof}

\begin{lemma}[Parallel ORAM]
  \label{lem:parallel-oram}
  The set of parallel physical data operations performed by the proxy ORAM executor over one epoch (as described in \S\ref{sec:parallel})
  is completely determined by the set of sequential physical accesses required to perform the same logical actions in Ring ORAM (plus a single write to the durability store).
\end{lemma}

\begin{proof}
  We note that, as described in \S\ref{sec:parallel}, the proxy performs all reads within an epoch before any writes (aside from the durability store).
  By construction, it ensures that each physical block that would be read at least once within an epoch in a fully sequential access is read exactly once in that epoch,
  and no other physical blocks are ever read (excluding crash recovery).

  This is enforced by holding a record of every block that has been read this epoch and then performing the reads of the sequential access, but skipping blocks that have already been read.
  Additionally, whenever an evict path operation would happen, the proxy reads every unread block from each bucket along that path, thus marking them as read.
  As the timing of evict paths is determined by how many data accesses have happened and their locations are deterministic,
  this enforcement mechanism is dependent only on the physical blocks accessed, not in any way on the data held in those blocks.

  Similarly, each block that would be written at least once in a sequentially-processed epoch is written exactly once at the end of the epoch.
  This is done by buffering writes in the proxy, allowing one buffered write of a physical block to overwrite any previous unflushed writes of that block.
  Then when the epoch ends, the proxy flushes all buffered writes.
  Again, the set of blocks being written is determined entirely by the physical access pattern of the sequential operation.

  Finally, a fixed amount of data is written to the durability store before each read batch, and the entire durability store is written with each write batch.
  This means that in normal operation, the location and timing of all reads and writes are determined by only the physical operations needed to perform the epoch operations sequentially and some extra completely deterministic operations.

  On crash recovery, the proxy reads the durability store and rereads all paths in the aborted epoch.
  This, again, is based entirely on physical access patterns.

  Hence all physical read and write operations within a parallelized epoch are determined entirely by the physical data operations needed to perform that epoch sequentially.
\end{proof}

\subsection{Proof of Security}

We now prove that the \sys protocol \prot (with access to \idealEpoch) is secure with respect to the ideal functionality described in Algorithm~\ref{alg:ideal-func}.
Let $\mathsf{Real}_{\adv,\env}(\lambda)$ denote the full transcript of \adv (including its inputs and randomness) when interacting with \prot.
Let $\mathsf{Ideal}_{\sdv,\env}(\lambda)$ denote the transcript produced by \sdv when run in the ideal world, interacting with \ideal.

\begin{theorem}
  Assume the encryption scheme used in \prot is semantically secure and the MACs are existentially unforgeable.
  For all probabilistic polynomial time (PPT) adversaries \adv and environments \env,
  there is a simulator \sdv[\adv] such that for all PPT distinguishers $\mathcal{D}$ there is some negligible function $\mathit{negl}$ such that
  \begin{align*}
    \Big| \Pr & \Big[\mathcal{D}\Big(\mathsf{Real}_{\adv,\env}(\lambda)\Big) = 1 \Big] \\
    & - \Pr\Big[\mathcal{D}\Big(\mathsf{Ideal}_{\sdv[\adv],\env}(\lambda)\Big) = 1 \Big] \Big| \leq \mathit{negl}(\lambda).
  \end{align*}
\end{theorem}

\begin{proof}
  This proof follows from a series of hybrid simulators, each of which is indistinguishable from the previous.

  We define hybrids $H_0,\dotsc, H_4$.
  $H_0$ operates in the real world with $\sdv[0]$ being a ``dummy'' that passes all messages through to \adv unmodified.
  $H_1$ has two ORAMs that are identical except for the MACs, one maintained by \adv and the other maintained by \sdv[1].
  $H_2$ replaces all data in \adv's ORAM with random dummy data, independent from the actual data.
  $H_3$ replaces the access pattern in \adv's ORAM with random data accesses.
  Finally $H_4$ uses \sdv[\adv] in the ideal world and no longer maintains its own ORAM.

  \textbf{Hybrid $H_0$} contains a dummy simulator that passes messages between \adv and the proxy unchanged.
  This produces a transcript identical to the real world.

  \textbf{Hybrid $H_1$} passes all messages through to \adv, but also maintains its own copy of the ORAM, simultaneously processes requests internally.
  On initialization \sdv[1] generates its own MAC key according to the same distribution as \prot's MAC key.
  It then replaces the MACs of all data sent to \adv with valid MACs on the same data using this new key.
  When \adv responds to a request, \sdv[1] checks the MACs on the data.
  If they are correct, it forwards the (correct) response from it's own ORAM with the original MACs.
  If they are incorrect, it responds with a failure message.
  If \adv's response is correct, so too will \sdv[1]'s.
  If \adv's MACs to not verify, \prot fails, so a failure message produces the same result.
  If \adv's response is wrong but the MACs verify, \adv must have forged a MAC since they include the data, position, and epoch counter, and no two pieces of data are ever given the same position and epoch counter.
  Moreover, because \prot has access to a trusted epoch counter via \idealEpoch, it can properly verify that the data has the correct epoch counter, even after crashes.
  Thus, if \sdv[1] accepts an incorrect response with non-negligible probability, we can simulate \adv to forge a MAC with non-negligible probability.
  Hence $H_1$ is computationally indistinguishable from $H_0$.

  Note that the MACs are only used to check that \adv provided correct data.
  If the storage server is assumed to be honest, this will always be the case and we can eliminate the MACs entirely (and also $H_0$ and $H_1$ become identical).

  \textbf{Hybrid $H_2$} replaces all data blocks provided to \adv with valid encryptions of random data and MACs on those encryptions.
  It otherwise passes on requests, including the location and timing of reads and writes.
  \sdv[2] continues to furnish responses to the proxy's queries using its internal ORAM with the original data, checking MACs according to the same scheme as in $H_1$.
  \sdv[2] then output's \adv's transcript.
  As all data is encrypted, the only difference between $H_1$ and $H_2$ is the contents of the ciphertexts, and by assumption the encryption scheme is semantically secure.
  This means $H_1$ and $H_2$ must be computationally indistinguishable.

  \textbf{Hybrid $H_3$} replaces all data requests to \adv with properly-formatted requests for randomly chosen data.

  When \sdv[3] receives a location log for a read batch, it logs an encryption of random (unrelated) data with \adv.
  When \sdv[3] receives the read instruction for a read batch, it first selects a random set of dummy paths of the batch size.
  It then requests \adv perform the proper parallel read operation for that dummy data.
  If \adv replies with the data and the MACs verify, \sdv[3] performs the actual reads on its separate ORAM with real data and returns the real data to the proxy.

  When \sdv[3] is notified of the end of an epoch and given the associated write batch,
  it determines which physical blocks to write using Ring ORAM's deterministic write sequence based on the total number of operations (both reads and writes) in an epoch.
  It then performs proper parallel writes of new encryptions of dummy data to each of those locations.
  If \adv replied with confirmed writes, \sdv[3] performs the originally-specified operations on its separate ORAM and confirms success to the proxy.

  Finally, if \sdv[3] receives a request to handle a proxy crash at epoch $c_e$ and batch $c_b$,
  it queries \adv as per the crash recovery protocol for that epoch and batch.
  When \adv provides valid (MAC-verifying) read path logs for any batches this epoch, \sdv[3] provides the associated logs to the proxy.
  When the proxy issues redo read requests, \sdv[3] issues the same requests it did the first time to \adv for the associated batches.
  Because \sdv[3] did not crash, it is able to retain which paths were read without having to store them explicitly.
  In is possible that the last read batch requested during recovery corresponds to a read that was never executed, in which case \sdv[3] generates a new random read batch and executes that instead.
  If \adv responds correctly, \sdv[3] responds to the proxy's requests.

  By Lemma~\ref{lem:cache-dedup}, the physical operations needed to process all real requests in a given epoch sequentially
  form an identical distribution to the sequential accesses needed to process the random requests chosen by \sdv[3].
  By Lemma~\ref{lem:parallel-oram}, applying the parallelization process relies only on the sequential physical access pattern,
  meaning it can be applied the same way to \sdv[3]'s random operations as to the real operations provided by the proxy.
  This means that the operations \sdv[3] requests of \adv are identically distributed to those the proxy requests of \sdv[3] when there are no crashes.

  When a crash occurs, the recovery procedure is guaranteed to reread all previously-read data, and any future reads must have independently random paths.
  This is because \sdv[3] does not even generate random paths to read until the read request is issued, by which point the persistent batch counter $c_b$ is updated.
  So if a crash does occur, it will redo any previous reads and future operations are treated as regular read/write batches with the same (independent) distribution.
  Since these are the only difference between $H_2$ and $H_3$, the two must produce identical distributions.

  \textbf{Hybrid $H_4$} now interacts with the ideal functionality and no longer maintains its own internal ORAM copy, only the data necessary to perform actions on \adv's, including the new MAC and encryption keys.
  The only data \sdv[3] was using to compute requests for \adv was the timing of batches and crash recoveries, and the epoch and batch counters during recovery.
  As \ideal explicitly provides all of that information, \sdv[4] is able to provide \adv with an identical view.
  Note that on crash recovery, this identical view requires completing a crash-recover epoch, which \sdv[4] can do by creating an appropriate number of read and write operations as it would in $H_3$.
  This means that $H_3$ and $H_4$ are identically distributed.

  Thus we see that $H_0$ corresponds to the real world, $H_4$ corresponds to the ideal world,
  and each sequential pair of $(H_i, H_{i+1})$ produce computationally indistinguishable transcripts.
  Thus it must be the case that $H_0$ and $H_4$ form computationally indistinguishable transcripts, so \prot realizes \ideal.
\end{proof}